\documentclass[
 reprint,
 amsmath,amssymb,
 aps,
 twocolumns
]{revtex4-2}

\usepackage{graphicx}
\usepackage{amsthm}
\usepackage{xcolor}
\usepackage{tikz}
\usetikzlibrary{shapes.geometric}

\DeclareRobustCommand{\chg}[1]{{#1}} 

\newtheorem{theorem}{Theorem}
\newtheorem{lemma}{Lemma}

\tikzset{
  diagram state node/.style={
    regular polygon,
    regular polygon sides=3,
    shape border rotate=0,
    draw,
    fill=black,
    minimum size=5.5pt,
    inner sep=0pt
  },
  diagram effect node/.style={
    regular polygon,
    regular polygon sides=3,
    shape border rotate=180,
    draw,
    fill=white,
    minimum size=5.5pt,
    inner sep=0pt
  },
  diagram state bond/.style={line width=0.85pt},
  diagram effect bond/.style={line width=0.85pt,densely dashed},
  diagram local wire/.style={line width=0.65pt},
  diagram contraction/.style={fill=black},
  diagram panel title/.style={font=\bfseries\small,align=center},
  diagram object label/.style={font=\small,align=center},
  diagram formula/.style={font=\small,align=center},
  diagram site label/.style={
    font=\tiny,
    text=gray!70!black,
    inner sep=0.3pt
  }
}

\newcommand{\diagramlocalstate}[2]{%
  \node[diagram state node] at (#1,-1.02) {};
  \draw[diagram local wire] (#1,-0.88)--(#1,0);
  \node[diagram object label,anchor=north] at (#1,-1.32) {$#2$};
}
\newcommand{\diagramlocaleffect}[2]{%
  \node[diagram effect node] at (#1,1.02) {};
  \draw[diagram local wire] (#1,0)--(#1,0.88);
  \node[diagram object label,anchor=south] at (#1,1.30) {$#2$};
}
\newcommand{\diagrambipstate}[3]{%
  \begingroup
  \pgfmathsetmacro{\diagrammidpoint}{(#1+#2)/2}
  \draw[diagram state bond]
    (#1,0) .. controls (#1,-0.62) and (#2,-0.62) .. (#2,0);
  \node[diagram object label,anchor=north]
    at (\diagrammidpoint,-0.56) {$#3$};
  \endgroup
}
\newcommand{\diagrambipeffect}[3]{%
  \begingroup
  \pgfmathsetmacro{\diagrammidpoint}{(#1+#2)/2}
  \draw[diagram effect bond]
    (#1,0) .. controls (#1,0.62) and (#2,0.62) .. (#2,0);
  \node[diagram object label,anchor=south]
    at (\diagrammidpoint,0.56) {$#3$};
  \endgroup
}
\newcommand{\diagramjunction}[1]{%
  \fill[diagram contraction] (#1,0) circle (1.05pt);
}
\newcommand{\diagramsitelabel}[2]{%
  \node[diagram site label,anchor=south west]
    at (#1+0.07,0.04) {$#2$};
}

\begin{document}

\title{Undecidability of consistency of generalized probabilistic theories}

\author{Serge Massar}
\email{serge.massar@ulb.be}
\affiliation{Laboratoire d’Information Quantique CP224, Universit\'e libre de Bruxelles (ULB), Avenue F. D. Roosevelt 50, 1050 Bruxelles, Belgium}

\begin{abstract}

Generalized probabilistic theories (GPTs) provide a unifying framework encompassing classical theories, quantum theories, as well as hypothetical alternatives. We investigate both the problem of extending a system with  a finite set of transformations, and the problem of adding to a translation-invariant set of systems a finite  set of entangled states and effects, plus all their images under the translation symmetry.
We show that determining whether such extensions are consistent with the axioms of GPTs is undecidable.
The source of the undecidability is that these finite extensions can generate infinitely many conditions which must be checked: iterating transformations may produce infinitely many new transformations; similarly, entangled states and effects may generate infinitely many new states via the analog of teleportation.
Our results show that extending GPTs to include dynamics or entanglement encounters fundamental computability obstructions, which can only be circumvented by introducing additional physical or mathematical assumptions.
\end{abstract}

\maketitle

\section{Introduction}

Generalized probabilistic theories (GPTs) \cite{
M63,L70,BL70,G73,M69,M74,HN74,H01,B07,CDAP10,H11,MM11,BW19} are a framework for formulating physical theories, including both classical and quantum mechanics, as well as alternatives not encountered in nature but of conceptual interest. Since their initial formulation, GPTs have become central to the study of foundations of quantum mechanics.  They are the basis for many of the modern axiomatic approaches to quantum theory, provide a framework for the study of nonlocal correlations stronger than those allowed by quantum mechanics \cite{B07,P14}, and have been used to show that many properties 
and concepts initially thought to be unique to quantum theory are in fact much more general, including information-theoretic principles such as no cloning and no broadcasting, foundational notions such as contextuality
\cite{S21,SW21}, aspects of  thermodynamics \cite{SW14,CG15} and resource theories, and computational complexity \cite{LB15,BBH19}. For reviews on GPTs, see Refs.~\cite{JH14,W18,P23}.

We show that the GPT framework exhibits a fundamental limitation: the consistency of certain natural extensions is undecidable. These include adding discrete-time dynamics, i.e., transformations, and including entangled states and effects. 
We show that
there  cannot exist an effective procedure for deciding whether such extensions  are compatible with the axioms of GPTs, and, in particular, whether all probabilities generated by the extended theory remain non-negative and bounded by one \cite{T36}. This result has important implications for the study of GPTs, their consequences, and their interpretation.
 
Our first result concerns transformations, which describe the evolution of a GPT in discrete time. They play the role of completely positive maps in quantum theory, or stochastic matrices in classical probability theory. We ask two questions. First, can a given finite list of candidate transformations be included in \emph{some} GPT? Second, given starting state and effect spaces, can the state space be enlarged to accommodate the proposed transformations while remaining compatible with the given effects? Both questions are undecidable, even though all the data used to describe an instance are finite. 
The origin of the undecidability is that composing transformations can produce infinitely many new transformations which were not in the initial set. The structure of the transformations generated under composition can be so complex that their consistency cannot be decided by a general algorithm.
The iterations of transformations are described by products of matrices. 
We leverage results  on the  acceptance probability of finite-state automata \cite{rabin1963,paz1971,B74,BMT77,blondel2003,H07,gimbert2010,Rote24}
and on unboundedness of matrix products \cite{BT00}  to prove these results.
The approach is quite similar to the undecidability of reachability problems in quantum theory (see in particular Ref.~\cite{WCPG11}), but here the undecidable question concerns the existence of a consistent extension of the specified data.

Our second result concerns entanglement. We study a one-dimensional infinite chain whose description repeats every two sites, so that the whole chain is specified by a finite unit cell. 
We ask two questions. First, 
given finitely many kinds of entangled states and effects, together with all their translations along the chain, can they be augmented with local states and effects so  as to constitute a consistent GPT? Second, given starting local state and effect spaces, can they be extended consistently to include a starting set of entangled states and effects? Both questions are undecidable, even though all the data used to describe an instance are finite. The second question is undecidable even if the entangled states and effects, and the local states and effects, are both separately consistent.

The two results are closely connected because the analog of teleportation \cite{BBCJPW93}  for GPTs, discussed for instance in Ref.~\cite{BBLW12}, can be viewed as a probabilistic transformation. 
The teleportation can transform states into new states, and so on  {\it ad infinitum}. Are these new states compatible with the axioms of GPTs (i.e., all outcome probabilities are non-negative and bounded by one), or does an inconsistency arise?
This is undecidable.

Our results provide an explanation for why the development of generalized probabilistic theories seems incomplete. On the one hand, a general theory of dynamics in GPTs is lacking. It would indeed be desirable to know what kinds of evolutions are compatible with different state and effect spaces. Our results show that there can be no general effective characterization deciding this compatibility, at least for discrete-time evolution.

Similarly, a general classification  of what kinds of entanglement are possible in GPTs is lacking. 
Recent work has resorted to numerical searches to find consistent multipartite GPTs which include specific entangled states and measurements \cite{DLG24}. Our work shows that there can be no general effective characterization deciding this compatibility, at least for the infinite translation-invariant systems considered here. 
Our work also suggests that numerical approaches as in Ref.~\cite{DLG24} will often be  inefficient, as finite versions of undecidable problems are often NP-hard \cite{Ketal23}.

It is likely that other questions concerning the consistency of GPTs are also  undecidable. These include the consistency of continuous-time evolution, and the inclusion of entangled states and effects for systems composed of a finite number of subsystems (rather than the infinite translation-invariant systems we study here). We leave this for future work.

Our result has important implications for the  future study of GPTs and their interpretation. Indeed, one should not draw far-reaching  conclusions from the study of GPTs without knowing whether the theories under consideration admit consistent  extensions  to time evolution and to the inclusion of entanglement.
An important direction for future work is to add  physical or mathematical hypotheses  to the GPT framework in order to make the consistency of these extensions decidable, or even better,
 efficiently decidable. Clarifying these conditions will be essential for determining which GPTs can meaningfully serve as candidates for physical theories.

\section{Generalized probabilistic theory}
\label{sec:GPT}

\subsection{Single systems}
\label{subsec:SingleSys}

We recall the formulation of generalized probabilistic theories.
The sets of unnormalized states, $C(\mathcal{S}) \subset \mathbb{R}^{d+1}$\chg{,} and unnormalized effects\chg{,} $C(\mathcal{E}) \subset \mathbb{R}^{d+1}$\chg{,} are assumed to be \emph{proper cones} in a $(d+1)$-dimensional real vector space, i.e.,\ convex, closed, pointed, and generating subsets of $\mathbb{R}^{d+1}$.
(Recall that a cone contains every non-negative multiple of each of its elements. ``Pointed'' means that the cone contains no nonzero vector together with its negative. Convexity reflects that probabilistic mixtures of states (effects) are valid states (effects). ``Generating'' means that differences of elements of the cone span the whole vector space.)
The fact that the cones are taken to be full dimensional is not restrictive (see the discussion in Ref.~\cite{SSD23}). 

The cones are mutually included in their dual cones:
\[
C(\mathcal{S}) \subseteq C(\mathcal{E})^*,
\qquad
C(\mathcal{E}) \subseteq C(\mathcal{S})^*.
\]
(The dual cone of a set consists of the vectors having a non-negative scalar product with every vector in that set.) The two inclusions above therefore say exactly that every allowed state--effect pairing is non-negative:
\begin{equation}
0 \le e^{\top} \omega,
\qquad \forall\, e \in C(\mathcal{E}),\ \forall\, \omega \in C(\mathcal{S})
\quad \text{(positivity)},
\label{eq:positivity}
\end{equation}
where $e^{\top} \omega$ denotes the natural Euclidean scalar product between effects and states. 
It is often assumed that $C(\mathcal{S})$ and $C(\mathcal{E})$ are dual to each other, but this is not necessary in full generality (see Refs.~\cite{JL13,KRF87}). We follow this more general approach.

There exists a distinguished effect $u \in \operatorname{int}C(\mathcal{E})$ called the \emph{unit effect}. The set of normalized states is 
\begin{equation}
\mathcal S=\{\omega\in C(\mathcal S):u^\top\omega=1\}.
\label{eq:unit}
\end{equation}
A measurement $M = \{e_1,\dots,e_n\}$ is a finite set of effects that sum to the unit effect:
\begin{equation}
\sum_{k=1}^n e_k = u 
\quad \text{(measurement normalization)}.
\label{eq:meas_norm}
\end{equation}

The set of physical effects is the order interval
\begin{equation}
\mathcal{E}
=
\left\{
e\in C(\mathcal{E}):u-e\in C(\mathcal{E})
\right\}
=
C(\mathcal{E})\cap\bigl(u-C(\mathcal{E})\bigr).
\label{eq:physical_effects}
\end{equation}
Equivalently, $e$ is physical precisely when it can occur in the binary
measurement $\{e,u-e\}$. Since $0,u\in C(\mathcal{E})$,
Eq.~\eqref{eq:physical_effects} implies that
$0,u\in\mathcal{E}$. It also implies that the physical effect space is
closed under complementation:
\[
e\in\mathcal{E}\quad\Longrightarrow\quad u-e\in\mathcal{E}.
\]
The definition in Eq.~\eqref{eq:physical_effects} also guarantees closure under coarse graining. Indeed,
if $M=\{e_1,\ldots,e_n\}$ is a measurement and
$I\subseteq\{1,\ldots,n\}$, then
\[
\sum_{i\in I}e_i\in C(\mathcal{E}),
\qquad
u-\sum_{i\in I}e_i
=
\sum_{i\notin I}e_i\in C(\mathcal{E}),
\]
and hence $\sum_{i\in I}e_i\in\mathcal{E}$ by
Eq.~\eqref{eq:physical_effects}, where the sum over the empty set is
understood to be zero.

Given a state $\omega$ and a measurement $M=\{e_k\}$, the probability of outcome $k$ is
\begin{equation}
P(e_k \mid \omega) = e_k^{\top} \omega.
\label{eq:prob}
\end{equation}
Conditions \eqref{eq:positivity}, \eqref{eq:unit}, and \eqref{eq:meas_norm} ensure that these probabilities are non-negative and sum to $1$.

A set of physical effects $\mathcal F\subseteq\mathcal E$ is
\emph{tomographically complete} if it distinguishes normalized
states: whenever
\[
e^\top\omega=e^\top\omega'
\qquad\text{for every }e\in\mathcal F,
\]
one has $\omega=\omega'$.  In the present finite-dimensional setting,
this is equivalent to
$\operatorname{span}(\mathcal F \cup \{u\}) =\mathbb R^{d+1}$.  Thus the
full-dimensional physical effect spaces assumed here are
tomographically complete.  Operationally, the coordinates of a state
can therefore be reconstructed from the probabilities of a suitable
finite spanning subset of effects.

We will write $\mathcal{S}_0$ and $\mathcal{E}_0$ for the starting sets of normalized states and physical effects. We assume that the states have $d$ independent directions once normalization is imposed, and that the effects span the full $(d+1)$-dimensional space. Equivalently, the cones generated by these sets are full dimensional. If the above axioms hold, we say that the pair $(\mathcal{S}_0,\mathcal{E}_0)$ is a single-party GPT.

\subsection{Transformations}
\label{subsec:Trans}

Transformations are linear maps
\begin{equation}
T : \mathbb{R}^{d+1} \to \mathbb{R}^{d+1}, \qquad \omega \mapsto T\omega,
\end{equation}
that map states to states: $T(\mathcal{S}) \subseteq \mathcal{S}$. Hence they preserve normalization:
\begin{equation}
u^{\top} T \omega = u^{\top} \omega = 1, \qquad \forall\, \omega \in \mathcal{S}.
\end{equation}
Since the normalized states affinely span the normalization hyperplane, this is equivalent to
\begin{equation}
u^{\top} T = u^{\top}.
\label{eq:preserve_u}
\end{equation}

It is useful to distinguish physical transformations from candidate
transformations.  If $T_1$ and $T_2$ are transformations of a fixed
GPT, then $T_2T_1(\mathcal S)\subseteq\mathcal S$ follows
automatically.  In the decision problems below, however, we are given candidate transformations whose matrices
are initially assumed only to preserve normalization,
$u^\top T_i=u^\top$.  We must decide whether there exists a state
space that is invariant under every finite product $T_w$ of these
candidate maps.  This is the source of the first undecidability result.

\subsection{Composite systems}
\label{subsec:Comp}

Composite systems are introduced using the tensor
product. The identification of the composite vector space with the tensor product of the individual system vector spaces follows from the natural assumption of local tomography \chg{(see}
Refs.~\cite{B07,JH14,JL13,KRF87,W92}\chg{)}.

Consider two systems $A$ and $B$ with state or effect cones
\[
C(\mathcal{S}^A),\,C(\mathcal{E}^A) \subset \mathbb{R}^{d_A+1},
\qquad
C(\mathcal{S}^B),\,C(\mathcal{E}^B) \subset \mathbb{R}^{d_B+1}.
\]
The composite system $AB$ is described by cones
\[
C(\mathcal{S}^{AB}),\,C(\mathcal{E}^{AB}) \subset \mathbb{R}^{(d_A+1)(d_B+1)},
\]
satisfying $C(\mathcal{S}^{AB}) \subseteq (C(\mathcal{E}^{AB}))^*$ and $C(\mathcal{E}^{AB}) \subseteq (C(\mathcal{S}^{AB}))^*$, together with the following consistency conditions.

\paragraph{Product states and effects.}
For any $\omega^A \in \mathcal{S}^A$, $\omega^B \in \mathcal{S}^B$,
\begin{equation}
\omega^A \otimes \omega^B \in C(\mathcal{S}^{AB}).
\label{eq:product_state}
\end{equation}
For any $e^A \in C(\mathcal{E}^A)$, $e^B \in C(\mathcal{E}^B)$,
\begin{equation}
e^A \otimes e^B \in C(\mathcal{E}^{AB}).
\label{eq:product_effect}
\end{equation}

A normalized bipartite state is \emph{separable} if it is a convex combination of product states and is \emph{entangled} otherwise. Analogously, an effect is separable if it lies in the cone generated by product effects, and is entangled otherwise. 

Two extreme spaces can be defined: the minimal tensor product contains only separable
states, while the maximal tensor product contains only
separable effects  \cite{B07}. Other theories such as quantum mechanics
contain both entangled states and entangled effects.

\paragraph{Unit effect and normalization.}
The composite unit effect factorizes:
\begin{equation}
u^{AB} = u^A \otimes u^B.
\label{eq:unitAB}
\end{equation}
Normalized bipartite states $\omega^{AB} \in \mathcal{S}^{AB}$ satisfy
\begin{equation}
(u^{AB})^{\top} \omega^{AB} = 1.
\label{eq:normAB_state}
\end{equation}

\paragraph{Measurements and joint probabilities.}
A bipartite measurement is a set $M^{AB} = \{e^{AB}_k\}$ with $\sum_k e^{AB}_k = u^{AB}$.
The probability of  obtaining  {outcome} $e_k^{AB}$ on the bipartite state $\omega^{AB}$
is given by
\begin{equation}
P(e_k^{AB}|\omega^{AB}) =(e_k^{AB})^{\top} \omega^{AB}.
\label{eq:jointprob}
\end{equation}

If $M^A = \{e^A_i\}$ and $M^B = \{e^B_j\}$ are
measurements on the two subsystems satisfying
$\sum_i e^A_i = u^A$ and $\sum_j e^B_j = u^B$,
then the product measurement $M=\{e^A_i \otimes e^B_j\}$
is normalized
$
\sum_{i,j} e^A_i \otimes e^B_j = u^{AB}
$
and hence a valid measurement.

\paragraph{Local tomography.}
Local tomography mentioned above is a statement about composite systems, distinct from the tomographic completeness of a single-system effect set. If two joint states give the same probabilities for every product effect $e_k^A\otimes e_l^B$, where the local effects range over tomographically complete sets, then the two joint states are equal.

\paragraph{Partial trace (marginalization).}
The marginal state of subsystem $A$ is obtained
by the action of the unit effect on subsystem $B$:
\begin{equation}
\omega^A = \bigl(\operatorname{id}_{A}\otimes (u^B)^\top\bigr)\omega^{AB},
\label{eq:partialtraceA}
\end{equation}
and similarly,
\begin{equation}
\omega^B = \bigl((u^A)^\top\otimes \operatorname{id}_{B}\bigr)\omega^{AB}.
\label{eq:partialtraceB}
\end{equation}
The marginals thus have the correct normalization:
\[
(u^A)^{\top} \omega^A = (u^B)^{\top} \omega^B = 1.
\]   

These operations play the role of the \emph{partial trace} in quantum theory. A consistent composite must satisfy $\omega^A\in\mathcal S^A$ and $\omega^B\in\mathcal S^B$. More generally, applying any physical effect on one subsystem of a bipartite state must produce a subnormalized state in the cone of the other subsystem.

\subsection{Notation}
\label{subsec:Notation}

  It is convenient to use a coordinate system in which $u$ is the first unit vector and the other  {orthonormal} basis  {vectors} are denoted $\phi_\mu$, $\mu=1,\ldots, d$:
\begin{equation}
u^{\top} u = 1,\quad
u^{\top} \phi_\mu = \phi_\mu^{\top} u = 0,\quad
\phi_\mu^{\top} \phi_{\mu'} = \delta_{\mu\mu'},
\label{eq:u1}
\end{equation}
\begin{equation}
u^{\top} = (1,0,\dots,0)  { =(1,0_{1\times d})}.
\end{equation}
States $\omega \in \mathcal{S}$ and effects $e \in \mathcal{E}$ are written as
\begin{equation}
\omega = \begin{pmatrix} 1 \\ v \end{pmatrix},
\qquad
e = \begin{pmatrix} p \\ g \end{pmatrix},
\quad
p \in \mathbb{R},\ v,g \in \mathbb{R}^{d}.
\label{eq:statesmu}
\end{equation}
Here $v$ is the state Bloch vector and $g$ comprises the last $d$ effect coordinates. The first component of an effect is $p=e^\top\omega_0$, where
\begin{equation}
\omega_0 = \begin{pmatrix} 1 \\  {0_{d\times 1}} \end{pmatrix}.
\end{equation}
The positivity condition Eq.~\eqref{eq:positivity} is $p+g^\top v\geq0$. If $\omega_0\in\mathcal S$, then $p\geq0$.

In the coordinates \eqref{eq:u1}, a transformation has block form
\begin{equation}
T = \begin{pmatrix}
1 &  {\boldsymbol{0}_{1\times d}} \\
s & R
\end{pmatrix},
\qquad s \in \mathbb{R}^{d},\ R \in \mathbb{R}^{d \times d}.
\label{eq:transformation_block}
\end{equation}
Here $R$ acts linearly on the Bloch vector and $s$ is an affine shift. For vanishing affine shifts,  {composition corresponds} to multiplying the $R$ blocks:
\begin{align}
& T_1=\begin{pmatrix}1& {\boldsymbol{0}_{1\times d}}\\[2pt] {\boldsymbol{0}_{d\times 1}}&R_1\end{pmatrix},\ 
T_2=\begin{pmatrix}1& {\boldsymbol{0}_{1\times d}}\\[2pt] {\boldsymbol{0}_{d\times 1}}&R_2\end{pmatrix} \nonumber\\
& \Longrightarrow\ 
T_2 T_1=\begin{pmatrix}1& {\boldsymbol{0}_{1\times d}}\\[2pt] {\boldsymbol{0}_{d\times 1}}&R_2 R_1\end{pmatrix}.
\label{Eq:ProdT}
\end{align}
With vanishing affine shifts, the Bloch block of an iterated
transformation is therefore given by a product of the matrices $R_i$.
This is the bridge to the matrix-product problems discussed below.

\subsection{Teleportation}
\label{subsec:Teleportation}

The analog of teleportation (remote state preparation) in GPTs is discussed  in Ref.~\cite{BBLW12}. Consider systems $A,B,C$, a state $\omega^A$ of $A$, a bipartite state $\omega^{CB}$ of $CB$, and a measurement $M^{BA}=\{e^{BA}_k\}$ acting jointly on $BA$. The initial state is $\omega^{CB} \otimes \omega^A$. Upon measuring $M^{BA}$, the (unnormalized) state of $C$ conditional  on outcome  {$k$} is
\begin{equation}
\omega^C_{\mid k}
= \bigl(\operatorname{id}_{C} \otimes (e^{BA}_k)^{\top} \bigr)\, \bigl(\omega^{CB} \otimes \omega^A\bigr).
\label{eq:telep_cond}
\end{equation}

Equation~\eqref{eq:telep_cond} takes a particularly simple form when the affine shifts vanish. Assume that $\omega^A$, $\omega^{CB}$\chg{,} and $e^{BA}_k$ take the form
\begin{align}
\omega^A &= u^A + \sum_{\mu} v^A_\mu\, \phi_\mu^A, \nonumber\\
e^{BA}_k &= p_k \!\left( u^B u^A + \sum_{\nu,\mu} H^{BA}_{\nu\mu;k}\, \phi_\nu^B \phi_\mu^A \right), \nonumber\\
\omega^{CB} &= u^C u^B + \sum_{\gamma,\nu} \Omega^{CB}_{\gamma\nu}\, \phi_\gamma^C \phi_\nu^B .
\end{align}
Here we impose the restriction that both
bipartite tensors have no one-body terms.  Thus $e_k^{BA}$ has no
components along $u^B\phi_\mu^A$ or $\phi_\nu^B u^A$, while
$\omega^{CB}$ has no components along $u^C\phi_\nu^B$ or
$\phi_\gamma^C u^B$.  Under this restriction the conditional map has
constant outcome weight and vanishing affine shift. We then have
\begin{equation}
\omega^C_{\mid k} 
= p_k \left( u^C + \sum_{\gamma} \Biggl( \sum_{\mu,\nu} \Omega^{CB}_{\gamma\nu}\, H^{BA}_{\nu\mu;k}\, v^A_\mu \Biggr) \phi_{\gamma}^C \right) ,
\end{equation}
which corresponds to an (unnormalized) transformation from $A$ to $C$ with vanishing affine  {shift} given by the product of matrices $\Omega^{CB}$ and 
$H_{; {k}}^{BA}$:
\begin{equation}
T_{\mid k} 
= p_k \begin{pmatrix}
1 &  {\boldsymbol{0}_{1\times d}} \\
 {\boldsymbol{0}_{d\times 1}} & \ \Omega^{CB} H^{BA}_{;k}
\end{pmatrix}.
\label{eq:telep_channel}
\end{equation}
Thus, for an outcome $k$ with $p_k>0$, the normalized conditional map acts on Bloch vectors as
$v^A\mapsto\Omega^{CB}H^{BA}_{;k}v^A$.  Its normalization is the
positive factor $p_k$.  Concatenating such teleportation steps therefore
multiplies the corresponding  matrices.

\subsection{Elementary systems}
\label{subsec:ElemSys}

In the case of multipartite systems it is useful to make a distinction between
specific states and effects and system types \chg{(see} Ref.~\cite{B07} for a discussion\chg{)}.

A system of type $A$ has certain properties, such as dimensionality {and sets of single-party states, effects, and transformations,} which are the same for all systems of type $A$. We can have multiple states of systems of type $A$, which we will denote $\omega^{A^{(1)}}, \omega^{A^{(2)}},\ldots$. 
Similarly we can have systems of other types, called $B$, $C$, etc. 

In quantum mechanics, a system type corresponds to a class of systems with  the same physical properties. For instance, type $A$ could correspond to a photon whose polarization can be measured and manipulated in certain ways, and type $B$ to an electron whose position can be measured and manipulated in certain ways. One can have multiple (distinguishable) photons of type $A$, and multiple (distinguishable) electrons of type $B$.

By using system types, we can have a succinct description of  {arbitrarily} many systems. These systems can be coupled together through multipartite effects and states, as in teleportation.

\subsection{Hypersphere theory and interior balls}

\subsubsection{Hypersphere theory}\label{ToyExample}

 {The undecidability reductions below require GPTs whose state and effect spaces contain neighborhoods of the completely mixed state and the half-unit effect, respectively, and whose positivity conditions can be verified directly. We adapt for this purpose a well-known single-system GPT, sometimes called hypersphere theory, whose Bloch-vector sets are Euclidean balls \chg{(see}, for instance, Refs.~\cite{BMU14,LPW18}\chg{)}.}

For $r>0$, let
\[
B_r^{d} := \{x \in \mathbb{R}^{d} : \lVert x \rVert \le r\}
\]
be the radius-$r$ ball in $\mathbb{R}^d$. For $p,r>0$,  denote  the  closed ball of radius $pr$ translated by $p$ along the first coordinate 
\begin{equation}
\Sigma_{p,r} := \left\{\, p \begin{pmatrix} 1 \\ x \end{pmatrix} \in \mathbb{R}^{d+1} : x \in B_r^{d} \,\right\} .
\end{equation}

\begin{lemma}[Hypersphere \chg{Theory}]
\label{lem:hypersphere-theory}
Take
\begin{align}
\mathcal{S}_\epsilon
&=\Sigma_{1,\epsilon},
\label{eq:genSepsilon}\\
\mathcal{E}_{\epsilon'}
&=
\operatorname{conv}
\left(
\Sigma_{1/2,\epsilon'}\cup\{0,u\}
\right),
\label{eq:genEepsilon}
\end{align}
where $\epsilon,\epsilon'>0$ satisfy
$\epsilon\epsilon'\leq1$.
The pair
$(\mathcal{S}_\epsilon,\mathcal{E}_{\epsilon'})$ is a
single-system GPT in the sense of Sec.~\ref{subsec:SingleSys}.
\end{lemma}

\begin{proof}
Every state has the form
\[
\omega=
\begin{pmatrix}1\\v\end{pmatrix},
\qquad
\lVert v\rVert\leq\epsilon,
\]
and is therefore normalized.  The set is compact and convex, is
full dimensional in the normalization hyperplane, and contains the
completely mixed state
$\omega_0=(1,\boldsymbol{0}_{d\times1})^{\top}$.
The conic hull is the set obtained by allowing non-negative multiples and sums. Here it is the proper cone:
\[
C(\mathcal{S}_\epsilon)
=
\left\{
\begin{pmatrix}t\\v\end{pmatrix}:
t\geq0,\ \lVert v\rVert\leq\epsilon t
\right\}.
\]

Let us check that the set $\mathcal{E}_{\epsilon'}$ is a valid effect space; i.e., it contains $0$ and
$u$ and is convex, full dimensional, and closed under complementation. The last property is the only one that needs explicit calculation.
To this end note that we
have the equivalent description
\begin{equation}
\mathcal{E}_{\epsilon'}
=
\left\{
\begin{pmatrix}p\\f\end{pmatrix}:
0\leq p\leq1,\ 
\lVert f\rVert\leq
\min\{\epsilon' p,\epsilon'(1-p)\}
\right\}
\label{eq:hypersphere_effect_interval}
\end{equation}
from which it follows that
\[
C(\mathcal{E}_{\epsilon'})
=
\left\{
\begin{pmatrix}p\\f\end{pmatrix}:
p\geq0,\ \lVert f\rVert\leq\epsilon'p
\right\},
\]
and therefore 
\[
\mathcal{E}_{\epsilon'} =
C(\mathcal{E}_{\epsilon'})
\cap
\bigl(u-C(\mathcal{E}_{\epsilon'})\bigr).
\]

It remains to verify positivity.  For
\(
\omega=(1,v)^{\top}\in\mathcal{S}_\epsilon
\)
and
\(
e=(p,f)^{\top}\in\mathcal{E}_{\epsilon'}
\),
Eq.~\eqref{eq:hypersphere_effect_interval} and the \chg{Cauchy-Schwarz}
inequality give
\[
e^{\top}\omega
=p+f^{\top}v
\geq
p-\lVert f\rVert\lVert v\rVert
\geq
p(1-\epsilon\epsilon')
\geq0.
\]
Since $u-e\in\mathcal{E}_{\epsilon'}$, applying the same inequality to
$u-e$ also gives
\[
e^{\top}\omega
\leq u^{\top}\omega=1.
\]
Hence all outcome probabilities lie in $[0,1]$, and their sum is $1$
for every measurement.  This proves the consistency of the
hypersphere theory.
\end{proof}

\subsubsection{Interior balls}

The following lemma records the geometric fact needed in the
reductions.  Every full-dimensional GPT contains a small ball of
normalized states about some interior state and a small ball of effects
about $u/2$.  If the state space is centrally symmetric under
$v\mapsto-v$, the state ball can be centered at the completely mixed
state $\omega_0$.

This is what lets the reductions work in both directions: bounded
products admit small ball models, whereas in the unbounded case every
full-dimensional GPT already contains test balls that eventually
detect a sufficiently large product.

\begin{lemma}[Interior \chg{Balls}]
\label{lem:hypersphere-Balls}
Let $(\mathcal{S},\mathcal{E})$ be any single-party GPT on $\mathbb{R}^{d+1}$ whose state and effect cones are full dimensional. Choose coordinates in which
$u=(1,\boldsymbol{0}_{1\times d})^\top$ and normalized states are
$\omega(v)=(1,v)^\top$.
The following two statements hold:
\begin{enumerate}
\item[(i)] There exist
$v_\ast\in\mathbb{R}^d$ and $\epsilon,\epsilon'>0$ such that
\begin{align}
\left\{
\begin{pmatrix}1\\v_\ast+x\end{pmatrix}:
\lVert x\rVert\leq\epsilon
\right\}
&\subseteq\mathcal{S},
\label{eq:state-open-set}\\
\Sigma_{1/2,\epsilon'}
&\subseteq\mathcal{E}.
\label{eq:effect-open-set}
\end{align}
\item[(ii)] 
Suppose that for any normalized state $\omega(v)=(1,v)^\top \in\mathcal{S}$, the state $\omega^-(v)=(1,-v)^\top$ also belongs to $\mathcal{S}$. Equivalently, suppose that
\[
J=
\begin{pmatrix}
1&\boldsymbol{0}_{1\times d}\\
\boldsymbol{0}_{d\times1}&-\mathbb{I}
\end{pmatrix}
\]
is a transformation that preserves the state space.
Then we can take $v_*=0$ in Eq. \eqref{eq:state-open-set}, or equivalently
there is $\epsilon >0$ such that
\[
\Sigma_{1,\epsilon}\subseteq\mathcal{S}.
\]
\end{enumerate}
\end{lemma}

\begin{proof}
For part~(i),
 the normalized state space has a nonempty relative
interior in the affine hyperplane $u^\top\omega=1$.  A sufficiently
small closed ball about a relative-interior state
$\omega(v_\ast)$ gives Eq.~\eqref{eq:state-open-set}.  Similarly,
$\mathcal{E}$ is full dimensional, convex, and invariant under
$e\mapsto u-e$.  If $e_0$ is an interior effect, then $u-e_0$ is
also an interior effect; by convexity, their midpoint $u/2$ is an
interior effect.  Hence $\mathcal{E}$ contains a sufficiently small
ball about $u/2$, which gives Eq.~\eqref{eq:effect-open-set}.

For part~(ii), the hypothesis implies that the Bloch-vector section of $\mathcal{S}$ is a convex set of dimension $d$ symmetric about the completely mixed state $\omega_0=(1,0_{1\times d})^\top$. Hence it
contains a Euclidean ball about
the origin.
\end{proof}

\section{Undecidable problems concerning matrix products}
\label{sec:undecidable}

We now recall two undecidable questions about products of matrices. They will be the engines of all four reductions in the paper. Given matrices $\{M_i\in\mathbb{R}^{d\times d}\}_{i=1}^k$, a word $w=i_1,\ldots,i_L$ tells us which matrices to apply and in which order, and defines the product 
\begin{equation}
M_w=M_{i_L}\cdots M_{i_1}. 
\end{equation}
Although the starting list is finite, words can have arbitrary length. 
The set of all such products\chg{,} $\{ M_w \}$\chg{,} is called 
the semigroup generated by the matrices. In our GPT constructions, such products arise either by composing transformations or by repeating teleportation [see Eqs.~\eqref{Eq:ProdT} and \eqref{eq:telep_channel}].

 Throughout this section, probability distributions are represented by column vectors. A stochastic matrix is therefore a matrix with non-negative entries whose columns sum to one, so that it maps probability distributions to probability distributions.

\emph{Acceptance probability of probabilistic automata.}
A probabilistic automaton is a finite-state machine whose transitions are probabilistic. Reading a letter applies one stochastic matrix; reading a word applies their product. The question below asks whether some word is accepted with probability above a fixed threshold.
Let $\{M_i\in [0,1]^{d\times d}\}$ with $i\in\{1,...,k\}$  be a set of $k$ stochastic matrices.
Let $q \in [0,1]^d$ be a probability distribution (i.e., $\sum_i q_i = 1$).
Let $F \in \{0,1\}^d$ be a $d$-dimensional vector with $0,1$ entries. 
Let $w=i_1,\ldots,i_L \in\{1,...,k\}^L$ be a word of length $L$.
The acceptance probability of word $w$ on input $q$ is
$P(F\vert w , q) =F^{{\top}} M_w q$.

There are many questions concerning the acceptance probabilities of probabilistic automata which are undecidable (see, e.g., Refs.~\cite{rabin1963,paz1971,B74,BMT77,blondel2003,H07,gimbert2010,Rote24}). We will use the following result:

\begin{theorem}\label{thm:emptiness}
Strict emptiness is undecidable.
Given $\{ M_i \in \mathbb{Q}^{d\times d}\}$, $q \in \mathbb{Q}^{d}$, $F \in \{0,1\}^d$ as defined above, and 
$ \lambda \in  \mathbb{Q} $  with $0 < \lambda <1$, 
it is undecidable whether there exists a word $w$ such that $P(F\vert w , q) > \lambda$, or whether $P(F\vert w , q) \leq  \lambda$ for all words $w$.
\end{theorem}

The two alternatives in this theorem are separated by the strict inequality: either at least one finite word crosses the cut point $\lambda$, or no word does. There is no algorithm which always decides which alternative holds.

Reference \cite{Rote24} gives the most recent results. In particular, it shows that Theorem \ref{thm:emptiness} holds for $k=2$ matrices, dimension $d=20$, and cut point $\lambda=1/20$; and for $k=5$ matrices, dimension $d=9$, and cut point $\lambda=1/9$. In both cases the initial and accepting vectors can be chosen to be distinct standard basis vectors, i.e.,  we may take $q=(1,0,\ldots,0)^\top$ and $F=(0,1,0,\ldots,0)^\top$.

We will also  use results on the unboundedness of products of matrices \cite{BT00}. We denote by $\| M \| $ the spectral norm:
\begin{eqnarray}
\| M \| &=& \text{largest singular value of $M$}\nonumber\\
&=& \max_{x\ne0} \frac{\| M  x\| }{\| x\| } \nonumber\\
&=& \max_{n,n' : \| n\| = \| n'\| =1}  n'^{\top} M  n
\label{eq:spectralnorm}
\end{eqnarray}
where $x,n,n'$ are  {$d$-dimensional} vectors, and we use the  {Euclidean} norm for vectors.

\begin{theorem}\label{thm:unboundedness}
Unboundedness of products of matrices is undecidable.
Given a set of matrices $\{ M_i \in \mathbb{Q}^{d\times d}\}$ with $i\in\{1,\ldots,k\}$, 
it is undecidable whether the products $M_{w} $ stay bounded for all words $w$, or whether the products can become unbounded.
That is, it is undecidable whether there exists $\Lambda >0$ such that, for all words $w$, $\| M_{w}  \| \leq \Lambda$, or whether, for all $\Lambda>0$, there exists a word $w$ such that $\| M_w  \| > \Lambda$.
\end{theorem}

Here ``bounded'' means that one common number $\Lambda$ bounds the norm of every finite product, regardless of its length. The theorem says that no algorithm can always decide whether such a common bound exists.

The proof in Ref.~\cite{BT00} is based on showing that the strict emptiness problem is reducible to unboundedness of matrix products. The reduction does not change  {the} dimension of the matrices, but increases their number by one. Hence Theorem \ref{thm:unboundedness} holds for \chg{three} matrices 
of dimension $20$ and \chg{six} matrices of dimension $9$. As shown in Ref.~\cite{BT00}, Theorem \ref{thm:unboundedness} also holds for \chg{two} matrices whose dimension can be taken to be $9\times6=54$. 

The rest of the paper transfers these two matrix questions to GPTs.

\section{Consistency of single-system generalized probabilistic theory is undecidable}
\label{Sec:ConsistenceSingleGPT}

We begin with a single system and consider the consistency of discrete-time evolution. A finite list of candidate transformations is given and represents the operations that may be performed in one time step. Once the list is given, every finite sequence of these operations must also be admitted. 
We consider two versions of the consistency question. In the first, only the candidate transformations are specified and we ask whether some state and effect spaces can accommodate them. In the second, the state and effect spaces are specified as well, and we ask whether the three ingredients are mutually compatible.

Every instance has a finite rational description.  In particular,
the unit effect and the proposed transformation matrices have rational
entries.  Whenever state and effect spaces are part of the input, they
should have finite semialgebraic descriptions over the rationals.  The
examples below are simpler still: they are Euclidean balls of rational
radius, sometimes enlarged by adjoining finitely many rational points
and taking the convex hull.

\subsection{Definitions and consistency conditions}

Recall that candidate transformations preserve normalization but need not preserve a state space.  A generating
set of candidate transformations on $\mathbb R^{d+1}$ is a finite set
\[
\mathcal T_0=\{T_i\}_{i=1}^K
\]
of linear maps satisfying $u^\top T_i=u^\top$.  This condition
guarantees preservation of normalization, but not yet preservation of
any state space.

The set of transformations obtained from all finite sequences of generators (the generated semigroup) is
\begin{align}
\mathcal{T}_{\rm gen}
=&
\left\{
T_w=T_{k_L}\cdots T_{k_1}\right. \nonumber\\
&\left.
:
L\geq 0,\ 
w=k_1,\ldots,k_L\in\{1,\ldots,K\}^{L}
\right\},
\end{align}
where the empty word, $L=0$, gives the identity transformation. Since there are infinitely many words, the set $\mathcal{T}_{\rm gen}$ can have infinite cardinality.

A \emph{generating set for a single-system GPT} is a triple
$(\mathcal{S}_0,\mathcal{E}_0,\mathcal{T}_0)$, where
$(\mathcal{S}_0,\mathcal{E}_0)$ is a single-party GPT and
$\mathcal{T}_0$ is a generating set of candidate transformations.

Sequential application of the transformations on the states in $\mathcal{S}_0$ generates new states. The state space thus generated is
\begin{equation}
\mathcal{S}_{\rm gen}
=
 {\overline{\operatorname{conv}}}
\left\{
T_w\omega:
\omega\in\mathcal{S}_0,\ 
T_w\in\mathcal{T}_{\rm gen}
\right\}.
\end{equation}
The starting state space $\mathcal S_0$ need not itself be invariant under the candidate transformations.
Taking the  {closed} convex hull ensures that $\mathcal{S}_{\rm gen}$ includes probabilistic mixtures and their
limits, as required of a state space.

\emph{A generating set for a single-system GPT $(\mathcal{S}_0,\mathcal{E}_0,\mathcal{T}_0)$ is consistent} if $(\mathcal{S}_{\rm gen},\mathcal{E}_0)$ is a single-party GPT.


Because $(\mathcal S_0,\mathcal E_0)$ is already a single-party
GPT and every candidate map preserves $u$, consistency is equivalent
to
\begin{equation}
e^{\top}T_w\omega\geq 0
\quad\text{for all }\omega\in\mathcal S_0,
\ e\in\mathcal E_0,\ T_w\in\mathcal T_{\rm gen}.
\label{eq:checkpositivity}
\end{equation}
To see this, note that
linearity and continuity extend this inequality to the closed convex
hull $\mathcal S_{\rm gen}$, while
$u^\top T_w=u^\top$ gives $u^\top\sigma=1$ for every
$\sigma\in\mathcal S_{\rm gen}$.

Furthermore, we also have that $\mathcal S_{\rm gen}$ is compact.  This follows from the fact that the
full-dimensional physical effect space is tomographically complete
and invariant under $e\mapsto u-e$; hence it contains a neighborhood of
$u/2$ (see Lemma~\ref{lem:hypersphere-Balls}).  Positivity against the
effects in this neighborhood bounds every Bloch component of a
generated state.  Thus $\mathcal S_{\rm gen}$ is bounded and, by
definition, closed.  It has affine dimension $d$ because it contains
$\mathcal S_0$.  Consequently its conic hull is a proper generating
cone.

\chg{The effects in $\mathcal E_0$ are starting effects; sequentially induced effects $T_{w'}^\top e$ may also be included, but need not be considered explicitly since positivity $e^\top T_{w'}T_w\omega\geq0$ follows from the same condition Eq.~\eqref{eq:checkpositivity}.}

\emph{A generating set of candidate transformations $\mathcal{T}_0$ is consistent} if there exist a normalized state space $\mathcal{S}_0$ of affine dimension $d$ and a full-dimensional physical effect space $\mathcal{E}_0$ such that $(\mathcal{S}_0,\mathcal{E}_0,\mathcal{T}_0)$ is consistent.

\subsection{Theorem \ref{Theo:undecidability1}: \chg{Consistency} of transformations}

\begin{theorem}
Given a fixed rational unit effect $u$ and a finite generating set of rational candidate transformations $\mathcal{T}_0=\{T_i\}_{i=1}^{K}$ satisfying $u^\top T_i=u^\top$, the problem of deciding whether $\mathcal{T}_0$ is consistent is undecidable.
\label{Theo:undecidability1}
\end{theorem}

If Theorem \ref{thm:unboundedness} holds for $k$ matrices of dimension
$d$, then Theorem \ref{Theo:undecidability1} holds for $K=k+1$
transformations acting on systems of dimension $d+1$.

The idea of the proof is simple. We place the matrices from Theorem~\ref{thm:unboundedness} in the lower-right blocks of the candidate transformations. If every matrix product is bounded, sufficiently small hypersphere state and effect spaces make all probabilities non-negative. If the products are unbounded, any full-dimensional state and effect spaces contain small balls, and a sufficiently large product stretches a state far enough to make some probability negative. The extra generator $\operatorname{diag}(1,-\mathbb I)$ supplies both signs of every Bloch vector and lets us center the state ball at the completely mixed state.

\begin{lemma}[Hypersphere completion]
\label{lem:hypersphere-open-sets}
Suppose that every transformation in
$\mathcal{T}_{\rm gen}$ has the block form
\begin{align}
T_w
&=
\begin{pmatrix}
1&\boldsymbol{0}_{1\times d}\\
\boldsymbol{0}_{d\times1}&R_w
\end{pmatrix},
\nonumber\\
\lVert R_w\rVert
&\leq\Lambda ,
\label{eq:bounded-block-transformations}
\end{align}
for some $\Lambda >0$.  Then $\mathcal{T}_0$ admits a consistent
completion starting from hypersphere state and effect spaces.
\end{lemma}

\begin{proof}
The unit is $u=(1,\boldsymbol{0}_{1\times d})^\top$. Take the hypersphere theory of
Sec.~\ref{ToyExample},
\[
\mathcal{S}_0=\Sigma_{1,\epsilon},
\qquad
\mathcal{E}_0=
\operatorname{conv}
\left(
\Sigma_{1/2,\epsilon'}\cup\{0,u\}
\right),
\]
where $\epsilon,\epsilon'>0$ satisfy
$\epsilon\epsilon'\leq \min\{1, 1/ \Lambda\}$.

Given the condition on $\epsilon\epsilon'$, the pair $(\mathcal{S}_0,
\mathcal{E}_0)$ is a single-party GPT (see Lemma~\ref{lem:hypersphere-theory}).

We have $u^\top T_w=u^\top$, so it remains only to check the
condition in Eq.~\eqref{eq:checkpositivity}.
For
$\omega=(1,\epsilon x)^\top$ and
$e=\frac12(1,\epsilon' y)^\top$, with
$\lVert x\rVert,\lVert y\rVert\leq1$, one has
\begin{align}
e^\top T_w\omega
&=
\frac12
\left(
1+\epsilon\epsilon' y^\top R_wx
\right)
\nonumber\\
&\geq
\frac12
\left(
1-\epsilon\epsilon'\Lambda
\right)
\geq0.
\label{eTq}
\end{align}
The zero and unit effects give probabilities $0$ and $1$,
respectively. Positivity extends to their convex hull by linearity.
\end{proof}

\begin{proof}[Proof of Theorem \ref{Theo:undecidability1}]
Given rational matrices
$\{M_i\in\mathbb{Q}^{d\times d}\}_{i=1}^{k}$, take the specified unit effect to be $u=(1,0,\ldots,0)^\top$, and define the set of candidate transformations to be
\begin{equation}
\mathcal{T}_0=
\left\{
\begin{pmatrix}
1&\boldsymbol{0}_{1\times d}\\
\boldsymbol{0}_{d\times1}&-\mathbb{I}
\end{pmatrix}
\right\}
\cup
\left\{
\begin{pmatrix}
1&\boldsymbol{0}_{1\times d}\\
\boldsymbol{0}_{d\times1}&M_i
\end{pmatrix}
\right\}_{i=1,\ldots,k}.
\label{eq:genT0-A}
\end{equation}
The first generator commutes with all the others.  Hence the
lower-right block of every generated transformation is either
$M_w$ or $-M_w$ (including $\pm\mathbb I$ when appropriate).

If the products $M_w$ are bounded, choose
$\Lambda > 0$ such that $\lVert M_w\rVert\leq\Lambda$ for every
word.  Lemma \ref{lem:hypersphere-open-sets}
then supplies state and effect spaces for which
$\mathcal{T}_0$ is consistent.

Conversely, suppose that the products $M_w$ are unbounded and that a
consistent completion nevertheless exists. Because $J=\operatorname{diag}(1,-\mathbb I)$ belongs to $\mathcal T_0$ and the generated state space is closed under every generated transformation, $J\mathcal S_{\rm gen}\subseteq\mathcal S_{\rm gen}$. Thus the Bloch-vector section of $\mathcal S_{\rm gen}$ is centrally symmetric about the completely mixed state. By 
\chg{part}~(ii) of
Lemma~\ref{lem:hypersphere-Balls} there
are $\epsilon,\epsilon'>0$ such that
\[
\Sigma_{1,\epsilon}\subseteq\mathcal{S}_{\rm gen},
\qquad
\Sigma_{1/2,\epsilon'}\subseteq\mathcal{E}_0.
\]
Choose a word $w$ for which
\(
\lVert M_w\rVert>(\epsilon\epsilon')^{-1}
\).
By Eq.~\eqref{eq:spectralnorm}, there are unit vectors $n,n'$ such
that
\(
n'^{\top}M_wn>(\epsilon\epsilon')^{-1}
\).
For the generated state and effect
\[
\omega=
\begin{pmatrix}1\\-\epsilon n\end{pmatrix}
\in\mathcal{S}_{\rm gen},
\qquad
e=
\frac12
\begin{pmatrix}1\\\epsilon'n'\end{pmatrix}
\in\mathcal{E}_0,
\]
we obtain
\[
e^{\top}T_w\omega
=
\frac12
\left(1-\epsilon\epsilon'n'^{\top}M_wn\right)
<0,
\]
contradicting consistency.  Thus $\mathcal{T}_0$ is consistent exactly
when the products $M_w$ are bounded.  Theorem
\ref{thm:unboundedness} proves the claim.
\end{proof}

\subsection{Theorem \ref{Theo:undecidability2}: \chg{Consistency} of a triple}

\begin{theorem}
Suppose we are given a finitely described single-party GPT
$(\mathcal S_0,\mathcal E_0)$ and a finite set of rational candidate
transformations $\mathcal T_0=\{T_i\}_{i=1}^K$. Under the promise that
$\mathcal T_0$ is consistent, it is undecidable whether the triple
$(\mathcal S_0,\mathcal E_0,\mathcal T_0)$ is consistent.
\label{Theo:undecidability2}
\end{theorem}

Thus, in contrast with Theorem~\ref{Theo:undecidability1},
undecidability persists even under the promise that the \chg{state-effect}
pair and the transformation set are separately consistent: it arises
from deciding whether they are mutually compatible.  Here consistency
of $\mathcal T_0$ means that there exists \emph{some} choice of state
and effect spaces $(\mathcal S',\mathcal E')$ such that
$(\mathcal S',\mathcal E',\mathcal T_0)$ is consistent.

 If Theorem
\ref{thm:emptiness} holds for $k$ matrices of dimension $d$, then
Theorem \ref{Theo:undecidability2} holds for $K=k$ transformations
acting on systems of dimension $d+1$.

This second proof follows the same \chg{matrix-transformation} correspondence, but uses stochastic matrices. Their products are automatically bounded, so the proposed transformations are consistent with some GPT. We then add one state determined by the initial probability distribution $q$ and two complementary effects determined by the accepting vector $F$. Every probability remains harmless except the pairing of this state with these effects after a word $w$. That pairing is non-negative exactly when the automaton's acceptance probability does not exceed $\lambda$.

\begin{lemma}[Uniform bounds for stochastic products]
\label{lemma3}
Let $M_w$ be any product of the column-stochastic matrices appearing in
Theorem~\ref{thm:emptiness}. Then
\[
\lVert M_w\rVert_1=1,
\qquad
\lVert M_w\rVert_\infty\leq d,
\qquad
\lVert M_w\rVert\leq\sqrt{d}.
\]
Moreover, for every probability vector $q$ and every
$F\in\{0,1\}^d$,
\[
\lVert M_wq\rVert\leq1,
\qquad
\lVert F^{\top}M_w\rVert\leq\sqrt{d}.
\]
\end{lemma}

\begin{proof}
Products of column-stochastic matrices are column stochastic. Their
maximum column sum is therefore $1$, while every row sum is at most
$d$. The spectral-norm inequality
\[
\lVert A\rVert
\leq
\sqrt{\lVert A\rVert_1\lVert A\rVert_\infty}
\]
then gives $\lVert M_w\rVert\leq\sqrt{d}$.
Since $M_wq$ is a probability vector,
\[
\lVert M_wq\rVert
\leq
\lVert M_wq\rVert_1
=1.
\]
Finally, for every $j\in\{1,\ldots,d\}$,
\[
0\leq (F^{\top}M_w)_j
=\sum_{i=1}^d F_i(M_w)_{ij}
\leq
\sum_{i=1}^d(M_w)_{ij}
=1.
\]
Consequently,
\[
\lVert F^{\top}M_w\rVert^2
=\sum_{j=1}^d (F^{\top}M_w)_j^2
\leq d,
\]
and hence $\lVert F^{\top}M_w\rVert\leq\sqrt{d}$.
\end{proof}

\begin{proof}[Proof of Theorem \ref{Theo:undecidability2}]

Let $\{M_i\}_{i=1,\ldots,k}$, $q$, $F$, and $0< \lambda <1$ be as in
Theorem~\ref{thm:emptiness}.  We restrict to the instances described
after that theorem, for which
\[
q=(1,0,\ldots,0)^{\top},
\qquad
F=(0,1,0,\ldots,0)^{\top}.
\]
Strict emptiness remains undecidable for this subclass.  Notice that
$F^{\top}q=0$; this will ensure that the additional state and effects
introduced below are compatible before any nontrivial transformation
is applied.

Define
\begin{equation}
\mathcal{T}_0=
\left\{
\begin{pmatrix}
1&\boldsymbol{0}_{1\times d}\\
\boldsymbol{0}_{d\times1}&M_i
\end{pmatrix}
\right\}_{i=1,\ldots,k}.
\label{eq:genT0}
\end{equation}
By Lemma~\ref{lemma3}, every generated lower-right block satisfies
$\lVert M_w\rVert\leq\sqrt d$.
Lemma~\ref{lem:hypersphere-open-sets}  therefore implies
that $\mathcal{T}_0$ is a consistent \chg{set of transformations}.

We next construct a single-party GPT to which these transformations
will be applied.  Define
\begin{align}
\mathcal{S}_0
&=
\operatorname{conv}
\left(
\Sigma_{1,\epsilon}
\cup
\left\{
\begin{pmatrix}1\\q\end{pmatrix}
\right\}
\right),
\label{eq:genS0}\\
e_F^\pm
&:=
\frac12
\begin{pmatrix}1\\ \pm F/\lambda\end{pmatrix},
\nonumber\\
\mathcal{E}_0
&=
\operatorname{conv}
\left(
\Sigma_{1/2,\epsilon'}
\cup
\left\{
\boldsymbol{0}_{(d+1)\times1},u,e_F^+,e_F^-
\right\}
\right),
\label{eq:genE0}
\end{align}
where
\begin{equation}
0<\epsilon'\leq\frac12,
\qquad
0<\epsilon\leq\frac{\lambda}{2\sqrt d}.
\label{eq:epsilonchoices}
\end{equation}

We first verify that $\mathcal{E}_0$ is a physical effect space in the
sense of Eq.~\eqref{eq:physical_effects}.  Let
\[
K=
\operatorname{conv}
\left(
B_{\epsilon'}^d
\cup
\left\{F/\lambda,-F/\lambda\right\}
\right).
\]
The set $K$ is compact, convex, and full dimensional, contains the origin,
and is centrally symmetric.  Equation~\eqref{eq:genE0} is equivalently
\begin{equation}
\mathcal{E}_0
=
\left\{
\begin{pmatrix}p\\z\end{pmatrix}:
0\leq p\leq1,\ 
z\in\min\{p,1-p\}K
\right\}.
\label{eq:E0_order_interval}
\end{equation}
Indeed, the section of $\mathcal{E}_0$ at $p=1/2$ is
$\{\frac12(1,f)^{\top}:f\in K\}$, and taking its convex hull with
$0$ and $u$ gives Eq.~\eqref{eq:E0_order_interval}.  Its conic hull is
\[
C(\mathcal{E}_0)
=
\left\{
\begin{pmatrix}p\\z\end{pmatrix}:
p\geq0,\ z\in pK
\right\}.
\]
This is a proper cone because $K$ is compact, \chg{is} full dimensional, and
contains the origin in its interior.
Since $K=-K$, Eq.~\eqref{eq:E0_order_interval} gives
\[
\mathcal{E}_0
=
C(\mathcal{E}_0)
\cap
\bigl(u-C(\mathcal{E}_0)\bigr).
\]
By Eq.~\eqref{eq:physical_effects}, this proves that
$\mathcal{E}_0$ is a valid physical effect space.
The state set $\mathcal{S}_0$ is compact and convex, is
full dimensional in the normalization hyperplane, and consists only of
normalized states.  Its conic hull is therefore also proper.  It
remains to check positivity between
$\mathcal{S}_0$ and $\mathcal{E}_0$.  We do this below together with
the positivity conditions for the generated states; taking the empty
word will establish that $(\mathcal{S}_0,\mathcal{E}_0)$ itself is a
single-party GPT.

For a word $w=i_1,\ldots,i_L$, write
$M_w=M_{i_L}\cdots M_{i_1}$ and
\[
T_w=
\begin{pmatrix}
1&\boldsymbol{0}_{1\times d}\\
\boldsymbol{0}_{d\times1}&M_w
\end{pmatrix}.
\]
For
\(
\omega=(1,v)^{\top}
\)
and
\(
e=(p,g)^{\top}
\),
\begin{equation}
e^{\top}T_w\omega
=p+g^{\top}M_wv.
\label{Eq:eTwomega}
\end{equation}
Because both sets in Eqs.~\eqref{eq:genS0} and \eqref{eq:genE0} are
convex hulls, it suffices to check positivity on the displayed
generating subsets.  Apart from the trivial effects $0$ and $u$, there
are four cases:
\begin{align}
\frac12\left(1+\epsilon\epsilon'x'^{\top}M_wx\right)
&\geq
\frac12\left(1-\epsilon\epsilon'\sqrt d\right)>0,
&&
\lVert x\rVert,\lVert x'\rVert\leq1,
\label{eq:PPP0}\\
\frac12\left(1+\epsilon'x'^{\top}M_wq\right)
&\geq
\frac12(1-\epsilon')>0,
&&
\lVert x'\rVert\leq1,
\label{eq:PPP1}\\
\frac12\left(1\pm
\frac{\epsilon}{\lambda}F^{\top}M_wx\right)
&\geq
\frac12\left(1-\frac{\epsilon\sqrt d}{\lambda}\right)>0,
&&
\lVert x\rVert\leq1,
\label{eq:PPP2}\\
(e_F^\pm)^{\top}T_w
\begin{pmatrix}1\\q\end{pmatrix}
&=
\frac12\left(1\pm
\frac1\lambda F^{\top}M_wq\right).
\label{eq:PPP3}
\end{align}
The first inequality follows from
$\lVert M_w\rVert\leq\sqrt d$; the second from
$\lVert M_wq\rVert\leq1$; and the third from
$\lVert F^{\top}M_w\rVert\leq\sqrt d$.
The choices in Eq.~\eqref{eq:epsilonchoices} make all three lower
bounds strictly positive.

For the empty word, $M_w=\mathbb I$.  Equations
\eqref{eq:PPP0}--\eqref{eq:PPP2} are therefore positive, while
Eq.~\eqref{eq:PPP3} equals $1/2$ because $F^{\top}q=0$.
By convexity, this proves positivity for every pair in
$\mathcal{S}_0\times\mathcal{E}_0$.  Together with the order-interval
description of $\mathcal{E}_0$, it proves that
$(\mathcal{S}_0,\mathcal{E}_0)$ is a single-party GPT.

For an arbitrary word, $F^{\top}M_wq\geq0$, so the plus sign in
Eq.~\eqref{eq:PPP3} is always positive.  The expression with the minus
sign is non-negative exactly when
\[
F^{\top}M_wq\leq\lambda.
\]
All other cases are already strictly positive, and convexity extends these checks to
all states and effects.  Consequently, the triple
$(\mathcal{S}_0,\mathcal{E}_0,\mathcal{T}_0)$ is consistent if and
only if $F^{\top}M_wq\leq\lambda$ for every word $w$.  This is the
strict-emptiness problem of Theorem~\ref{thm:emptiness}, and is
therefore undecidable.

\end{proof}

\section{Consistency of translation-invariant multipartite generalized probabilistic theory is undecidable}
\label{sec:multipartite-undecidability}

The aim of this section is to reproduce the two single-system undecidability results without putting transformations into the theory by hand. Instead, entangled states and effects will implement the transformations indirectly through teleportation. A measurement outcome contributes an overall positive factor to the resulting state, but this factor does not affect whether a later probability is positive or negative. What matters is the matrix acting on the Bloch vector, and repeated teleportation multiplies these matrices just as repeated time evolution did in the previous section.

The main difficulty is to prevent the entangled generators from combining in unintended ways that would introduce additional consistency conditions. We solve this by placing them on alternating bonds of a one-dimensional chain. The pattern repeats every two sites, but the sites retain distinct labels. This forces every connected experiment to follow a path along the chain, making it possible to list all probabilities that must be checked. We first define precisely what is meant by such a translation-invariant generating set, then state the two results, give the common construction used in both proofs, and finally give the proofs.

\subsection{Period-two translation-invariant chain}
\label{sec:1DTrans}

We now describe the class of multipartite systems used in the two
reductions.  The systems are arranged on an infinite one-dimensional
chain with sites labeled by $n\in\mathbb Z$.  Site $n$ carries an
elementary system $A^n$.  The chain has translation symmetry with period $2$.
In the following, we generalize the definitions of
Sec.~\ref{subsec:Comp} to this setting.

Let $\Lambda\subset\mathbb Z$ be a finite set of sites.
Multipartite states and effects with support on $\Lambda$ are
tensors with one tensor factor for each site in $\Lambda$.  Joint
probabilities, marginals, and conditional states or effects are
obtained by taking tensor products and contracting tensor indices.
Importantly, site labels must be respected: a tensor acting
on $A^n$ can be contracted only with another tensor carrying the same
site label.

We require that all even sites have one local GPT
and all odd sites another. Translation invariance then imposes
\[
(\mathcal S_0^{A^{n+2}},\mathcal E_0^{A^{n+2}})
\simeq
(\mathcal S_0^{A^n},\mathcal E_0^{A^n}).
\]

The translation symmetry is important because it allows us to provide a
finite description of the infinite chain. Thus, given a state or effect,
translation symmetry supplies
copies of that state or effect at translated positions. The
translation symmetry does not allow arbitrary relabelings of the tensor
indices. It is this restriction that prevents the entangled generators
below from being combined in unwanted ways.

If a finite-support tensor $v$ acts on sites
$\Lambda\subset\mathbb Z$, let $\tau(v)$ denote the tensor with the
same coefficients acting on $\Lambda+2$. From one seed tensor we
obtain the orbit
\[
\operatorname{Orb}_2(v)=\{\tau^k(v):k\in\mathbb Z\}.
\]
A translation-invariant family is \emph{finitely generated up to
translations} if it is the union of the orbits of finitely many seed
tensors or sets of tensors, each having a finite description. The input is therefore finite even though the chain and the
translated family are infinite.

\emph{A multipartite generating family} consists of the two local
GPTs (one for the even sites and one for the odd sites), finitely many
finite-support multipartite seed states, and finitely many
finite-support multipartite seed measurements, together with all of
their translates. Each seed measurement is required to sum to the
product unit effect on its support.

We denote by
$\mathcal{S}_0^{\rm multi}$ the multipartite seed states 
and all their translates.
We denote by $\mathcal{E}_0^{\rm multi}$ the multipartite
seed effects, and all their translates.

A \emph{finite experiment} is built recursively from the single-site and multipartite generators by tensor products and contractions. Tensor products of states, or of effects, require disjoint open supports. A state leg and an effect leg may be contracted only when they carry the same site label. When composing previously generated tensors, only
their open legs are available for further contraction.
Indices already contracted within each tensor remain
separate summation indices, even when they carry the
same site label.
After contraction, open state legs define an unnormalized conditional state; analogously, open effect legs define an unnormalized conditional effect.
A closed network gives the scalar weight of an outcome.
Disconnected networks represent tensor products, and contraction with
unit effects gives marginals.

Consider a multipartite generating family.
For each finite set of sites $\Lambda$, let
$C(\mathcal S_{\rm gen}^{\Lambda})$ be the closed cone generated by all
conditional states on $\Lambda$ obtained from finite experiments and
their marginals. Define
$C(\mathcal E_{\rm gen}^{\Lambda})$ analogously from the conditional
effects generated by finite experiments, including their
coarse grainings.
The normalized state space is the section
\[
\mathcal S_{\rm gen}^{\Lambda}
=\{\sigma\in C(\mathcal S_{\rm gen}^{\Lambda}):
(u^\Lambda)^\top\sigma=1\},
\]
where $u^\Lambda$ is the factorized unit effect
\[
u^\Lambda=\bigotimes_{n\in\Lambda}u^n.
\]

The generating family is \emph{consistent} if, for every finite
$\Lambda$, the generated state and effect cones are proper and
generating, are mutually positive, \chg{and} contain the required product
cones.
The specified measurements must sum to the corresponding unit effects.
Consequently, every closed network has non-negative weight, and the
outcome probabilities of every measurement sum to one.

For the period-two construction below, Lemma~\ref{lem:path-completion} 
shows that these requirements reduce to non-negativity of
four families of path probabilities.
\emph{A pair
$(\mathcal S_0^{\rm multi},\mathcal E_0^{\rm multi})$ is consistent} if
some period-two choice of local GPTs completes it in this sense.

\subsection{Statements of the results}

The two statements parallel Theorems~\ref{Theo:undecidability1} and \ref{Theo:undecidability2}. In Theorem~\ref{Theo:undecidability3}, the local GPTs are not part of the input: we ask whether any period-two local completion exists, and we show that, for the instances used in the reduction, whenever a completion exists, hypersphere GPTs can serve as starting local spaces. In Theorem~\ref{Theo:undecidability4}, the two local GPT descriptions are part of the finite rational input\chg{. The two local GPTs} and the multipartite generators are promised to be separately consistent. We ask whether the given local GPTs and multipartite generators are consistent when combined.

\begin{theorem}
Consistency of translation-invariant multipartite generators
is undecidable.
Given period-two translation-invariant families of multipartite
state and effect generators specified by finitely many rational
finite-support seed tensors and all their translates,
$(\mathcal{S}_0^{\rm multi},\mathcal{E}_0^{\rm multi})$, it is
undecidable whether there exist period-two translation-invariant
single-party GPTs that complete them to a consistent multipartite
generating set.
\label{Theo:undecidability3}
\end{theorem}

\begin{theorem}
Compatibility with specified starting single-party GPTs is undecidable.
Given finitely described, period-two translation-invariant
single-party GPTs
\[
\bigl\{
(\mathcal{S}_0^{A^n},\mathcal{E}_0^{A^n})
\bigr\}_{n\in\mathbb{Z}}
\]
and a consistent, finitely generated up to translations, period-two
translation-invariant pair of multipartite generator families
$(\mathcal{S}_0^{\rm multi},\mathcal{E}_0^{\rm multi})$, it is
undecidable whether their combination is a consistent multipartite
generating set.
\label{Theo:undecidability4}
\end{theorem}

Thus, in Theorem~\ref{Theo:undecidability4}, the single-party GPTs and
the multipartite generators are separately consistent.  The
undecidable question is whether they are mutually compatible.

Although the chain contains infinitely many sites, the
 inputs are finite: one specifies the even and odd local
dimensions, the $k$ bipartite state seeds, and  the two effects of one
bipartite measurement.  Translation supplies all copies along the
chain.  In Theorem~\ref{Theo:undecidability4}, one additionally
specifies the two local GPT descriptions repeated on the even and odd
sites.  

If the relevant matrix problem involves $k$ matrices of dimension
$d$, the construction uses $k$ seed multipartite state
generators and the two outcomes of one multipartite measurement up to
translations.

\subsection{Period-two chain that simulates matrix products}

Both undecidability proofs use the same chain, so we describe it once before giving the proofs. The construction has two ingredients. On one set of alternating bonds we place bipartite states labeled by matrices $\Omega_l$. On the other bonds we place a two-outcome bipartite measurement, with outcomes labeled $+$ and $-$. Contracting one state with one measurement teleports a local state by two sites and applies the matrix $\pm\Omega_l$ to its Bloch vector. Repeating the procedure along the chain produces the product $\pm\Omega_w$, for any word $w$.

The even and odd sites play different directional roles. A state placed on an even site can move only to the right, while a state on an odd site can move only to the left. This is not an additional physical assumption; it follows from which neighboring pairs carry state generators and which carry effect generators. The alternating geometry is what keeps the construction under control: every connected experiment is a path, and every path contains one well-ordered matrix product.

\chg{The reader may therefore keep the following picture in mind while reading the formulas. The bipartite states store the matrices, the bipartite measurements join consecutive steps, and the local state or effect at each end supplies the vectors on the two sides of the matrix product.}

\paragraph{Local coordinates and generator geometry.}
All even sites have one single-party GPT and all odd sites have
another, with the same vector-space dimension $d+1$.  Write their
normalized states and physical effects as
\begin{align}
\omega^{2n}(v)
&=
\begin{pmatrix}1\\v\end{pmatrix},
&
e^{2n}(p,g)
&=
\begin{pmatrix}p\\g\end{pmatrix},
\nonumber\\
\omega^{2n+1}(v')
&=
\begin{pmatrix}1\\v'\end{pmatrix},
&
e^{2n+1}(p',g')
&=
\begin{pmatrix}p'\\g'\end{pmatrix}.
\label{Eq:local-even-odd}
\end{align}

On every bond $(2n+1,2n)$, introduce the two effects
\begin{equation}
e_\pm^{2n+1,2n}
=
\frac12
\left(
u^{2n+1}u^{2n}
\pm
\sum_{\mu=1}^d
\phi_\mu^{2n+1}\phi_\mu^{2n}
\right).
\label{Eq:E2}
\end{equation}
They form a measurement because
\[
e_+^{2n+1,2n}+e_-^{2n+1,2n}
=u^{2n+1}u^{2n}.
\]

Read from left to right along the chain, these measurement effects occupy the \chg{even-odd} bonds. The opposite, \chg{odd-even} bonds will carry the bipartite state generators. Thus state and effect bonds alternate along the chain and never compete for the same role.

Let $\{\Omega_l\}_{l=1}^k$ be a finite set of $d\times d$ matrices labeled by $l$.  On every
complementary bond $(2n+2,2n+1)$, introduce the states
\begin{equation}
\omega_l^{2n+2,2n+1}
=
u^{2n+2}u^{2n+1}
+
\sum_{\mu,\nu=1}^d
(\Omega_l)_{\mu\nu}
\phi_\mu^{2n+2}\phi_\nu^{2n+1}.
\label{Eq:S2}
\end{equation}
Equations~\eqref{Eq:E2} and \eqref{Eq:S2}, together with all their
translations by an even number of sites, define the multipartite
generator families used in both reductions.

Only the label $l$, and hence the matrix $\Omega_l$, varies among the bipartite state generators. Although there are infinitely many translated copies along the chain, there are only $k$ seed state generators and two effect outcomes in the finite description.

The systems $A^n$ retain distinct labels even though their local
descriptions repeat with period two.  Consequently, the only
contractions between the bipartite generators are those permitted by
their displayed site labels.

\begin{figure*}[t]
\centering
\resizebox{\textwidth}{!}{%
\begin{tikzpicture}[x=1.05cm,y=1.05cm]
  \begin{scope}
    \node[diagram panel title] at (7.2,1.75)
      {(a)};

    \diagramsitelabel{0.7}{n}
    \diagramlocalstate{0.7}{\omega(v)}
    \node[diagram object label] at (0.7,-2.00) {single state};

    \diagramsitelabel{4.0}{n}
    \diagramlocaleffect{4.0}{e(p,\chg{f})}
    \node[diagram object label] at (4.0,-2.00) {single effect};

    \diagramsitelabel{6.5}{2n+1}
    \diagramsitelabel{8.0}{2n+2}
    \diagrambipstate{6.5}{8.0}{\omega_l(\Omega_l)}
    \node[diagram object label] at (7.25,-2.00)
      {bipartite state};

    \diagramsitelabel{10.2}{2n}
    \diagramsitelabel{11.7}{2n+1}
    \diagrambipeffect{10.2}{11.7}{e_s,\ s=\pm1}
    \node[diagram object label] at (10.95,-2.00)
      {bipartite effect\\(factor $1/2$)};

    \draw[->,line width=0.7pt] (13.7,-0.85)--(13.7,0.85);
    \node[diagram object label,rotate=90] at (14.05,0)
      {contraction direction};
  \end{scope}

  \begin{scope}[yshift=-5.1cm]
    \node[diagram panel title] at (7.2,2.15)
      {(b)};

    \foreach \x/\lab in {
      0.5/{2n},1.7/{2n+1},2.9/{2n+2},
      4.1/{2n+3},5.3/{2n+4}
    }{\diagramsitelabel{\x}{\lab}}
    \diagramlocalstate{0.5}{\omega^{2n}(v)}
    \diagrambipeffect{0.5}{1.7}{e_{s_1}}
    \diagrambipstate{1.7}{2.9}{\omega_{l_1}(\Omega_{l_1})}
    \diagrambipeffect{2.9}{4.1}{e_{s_2}}
    \diagrambipstate{4.1}{5.3}{\omega_{l_2}(\Omega_{l_2})}
    \foreach \x in {0.5,1.7,2.9,4.1}{\diagramjunction{\x}}
    \draw[diagram local wire] (5.3,0)--(5.3,0.78);
    \node[diagram formula,anchor=south] at (5.3,0.92)
      {$\widetilde{\omega}^{\,2n+4}_{s,w}(v)$};
    \draw[->,line width=0.7pt] (0.7,-1.82)--(5.1,-1.82);
    \node[diagram formula,anchor=north] at (2.9,-1.92)
      {even state propagates to the right};

    \foreach \x/\lab in {
      9.1/{2n+1},10.3/{2n+2},11.5/{2n+3},
      12.7/{2n+4},13.9/{2n+5}
    }{\diagramsitelabel{\x}{\lab}}
    \draw[diagram local wire] (9.1,0)--(9.1,0.78);
    \node[diagram formula,anchor=south] at (9.1,0.92)
      {$\widetilde{\omega}^{\,2n+1}_{s,w}(v')$};
    \diagrambipstate{9.1}{10.3}{\omega_{l_1}(\Omega_{l_1})}
    \diagrambipeffect{10.3}{11.5}{e_{s_1}}
    \diagrambipstate{11.5}{12.7}{\omega_{l_2}(\Omega_{l_2})}
    \diagrambipeffect{12.7}{13.9}{e_{s_2}}
    \diagramlocalstate{13.9}{\omega^{2n+5}(v')}
    \foreach \x in {10.3,11.5,12.7,13.9}{\diagramjunction{\x}}
    \draw[->,line width=0.7pt] (13.7,-1.82)--(9.3,-1.82);
    \node[diagram formula,anchor=north] at (11.5,-1.92)
      {odd state propagates to the left};
  \end{scope}
\end{tikzpicture}%
}
\caption{\label{fig:notation-directional}
Graphical conventions and directional teleportation.
(a) Solid triangles denote single-party states, open inverted
triangles single-party effects, solid cups bipartite state generators,
and dashed caps bipartite effect outcomes.  State generators occupy
\chg{odd-even} bonds, whereas effect outcomes occupy \chg{even-odd} bonds.
(b) Representative two-step contractions.  An even state propagates
to the right and acquires the ordered product $\Omega_w$; an odd state
propagates to the left across the same state bonds and acquires
$\Omega_w^{\top}$.  Black dots mark contractions, and the end point
without a dot is the open wire carrying the conditional output state.
The general $L$-step statement is
Lemma~\ref{lem:directional-teleportation}.}
\end{figure*}

The single-party and bipartite states and effects can be represented graphically as in Fig.~\ref{fig:notation-directional}(a).

\begin{lemma}[Directional teleportation]
\label{lem:directional-teleportation}
Consider a chain segment containing $L\geq 0$ consecutive state bonds,
read from left to right.  Its $j$th state bond occupies
$(2n+2j,2n+2j-1)$ and carries a label
$l_j\in\{1,\ldots,k\}$, for $j=1,\ldots,L$.  Define the word
$w=l_1,\ldots,l_L\in\{1,\ldots,k\}^L$, write $|w|=L$, and set
\begin{equation}
\Omega_w
=
\Omega_{l_L}\cdots\Omega_{l_1}.
\label{eq:Omega-word}
\end{equation}
For the empty word $\epsilon$, set
$|\epsilon|=0$ and $\Omega_\epsilon=I_d$.
Starting with an even state $\omega^{2n}(v)$ at the left end point and
contracting toward the right, or with an odd state
$\omega^{2n+2L+1}(v')$ at the right end point and contracting toward
the left, produces the respective subnormalized conditional states
\begin{align}
\widetilde{\omega}^{\,2n+2L}_{s,w}(v)
&=
\frac{1}{2^L}
\begin{pmatrix}
1\\s\,\Omega_wv
\end{pmatrix},
\label{eq:even-teleportation}\\
\widetilde{\omega}^{\,2n+1}_{s,w}(v')
&=
\frac{1}{2^L}
\begin{pmatrix}
1\\s\,\Omega_w^{\top}v'
\end{pmatrix},
\label{eq:odd-teleportation}
\end{align}
where $s\in\{+1,-1\}$ is the product of the measurement-outcome
signs along the contraction path\chg{,} for $L=0$, $s=+1$.
\chg{Thus even states propagate to the right, whereas odd states propagate to the left.}
\end{lemma}

Lemma~\ref{lem:directional-teleportation} is the key link with the matrix problems. It says that a path of $L$ teleportation steps does nothing more complicated than multiply the input Bloch vector by the ordered product $\Omega_w$, or by its transpose when the path is traversed in the opposite direction. The factor $2^{-L}$ is the probability weight contributed by the $L$ measurement outcomes and is always positive. The directional teleportation is represented graphically in Fig.~\ref{fig:notation-directional}(b).

\begin{proof}
For one step, direct contraction of
Eqs.~\eqref{Eq:E2} and \eqref{Eq:S2} gives
\[
\begin{pmatrix}1\\v\end{pmatrix}
\longmapsto
\frac12
\begin{pmatrix}1\\\pm\Omega_l v\end{pmatrix}
\]
for an even input and
\[
\begin{pmatrix}1\\v'\end{pmatrix}
\longmapsto
\frac12
\begin{pmatrix}1\\\pm\Omega_l^{\top}v'\end{pmatrix}
\]
for an odd input.  Iteration multiplies the normalization by $1/2$ at
each step.  Reading the state bonds from left to right gives the
ordered product in Eq.~\eqref{eq:Omega-word}; propagation in the
opposite direction gives its transpose.
\end{proof}

\begin{lemma}[Classification of finite outcome probabilities]
\label{lem:probability-classification}
Every connected closed finite experiment constructed from
Eqs.~\eqref{Eq:local-even-odd}--\eqref{Eq:S2} determines a word $w$:
the labels of all its state bonds, read from left to right.  Writing
$L=|w|\geq 0$ and using Eq.~\eqref{eq:Omega-word}, its outcome probability
has one of the following four forms:
\begin{align}
P^1
&=
\frac{p+s\,g^{\top}\Omega_wv}{2^L},
\label{Eq:P1}\\
P^2
&=
\frac{p'+s\,v'^{\top}\Omega_wg'}{2^L},
\label{Eq:P2}\\
P^3
&=
\frac{pp'+s\,g^{\top}\Omega_wg'}{2^{L-1}},
\label{Eq:P3}\\
P^4
&=
\frac{1}{2^{L+1}}
\left(1+s\,v'^{\top}\Omega_wv\right).
\label{Eq:P4}
\end{align}
Here $s\in\{+1,-1\}$ is the product of all bipartite
measurement-outcome signs on the path.  More precisely:
\begin{enumerate}
\item $P^1$ and $P^2$ contain $L$ state bonds and $L$ effect
bonds.
\item $P^3$ contains $L\geq1$ state bonds and $L-1$ effect bonds.
\item $P^4$ contains $L$ state bonds and $L+1$ effect bonds.
\end{enumerate}
When a path contains no bipartite effect outcome, $s=+1$.
The four cases are illustrated in
Fig.~\ref{fig:probability-cases}.
Every disconnected closed finite experiment is a product of probabilities
of these four forms.  Coarse-grained outcome probabilities are finite
sums of such products.
\end{lemma}

The four expressions have a simple interpretation. In $P^1$ and $P^2$, a local state enters one end of a path and a local effect closes the other. In $P^3$, local effects close both ends of a path that contains at least one bipartite state. In $P^4$, local states enter both ends and the path contains at least one bipartite effect. These are the only possibilities because the alternating bond geometry forces each connected component to follow a path with strictly monotone site labels.

\begin{figure*}[t]
\centering
\resizebox{\textwidth}{!}{%
\begin{tikzpicture}[x=1.05cm,y=1.05cm]
  \begin{scope}[xshift=0.4cm]
    \node[diagram panel title] at (2.4,2.00)
      {(a)};
    \foreach \x/\lab in {
      0/{2n},1.2/{2n+1},2.4/{2n+2},
      3.6/{2n+3},4.8/{2n+4}
    }{\diagramsitelabel{\x}{\lab}}
    \diagramlocalstate{0}{\omega(v)}
    \diagrambipeffect{0}{1.2}{e_{s_1}}
    \diagrambipstate{1.2}{2.4}{\omega_{l_1}(\Omega_{l_1})}
    \diagrambipeffect{2.4}{3.6}{e_{s_2}}
    \diagrambipstate{3.6}{4.8}{\omega_{l_2}(\Omega_{l_2})}
    \diagramlocaleffect{4.8}{e(p,g)}
    \foreach \x in {0,1.2,2.4,3.6,4.8}{\diagramjunction{\x}}
    \node[diagram formula] at (2.4,-2.05)
      {$\displaystyle P^1=\frac{1}{2^2}
      \left(p+s_1s_2\,g^\top
      \Omega_{l_2}\Omega_{l_1}v\right)$};
  \end{scope}

  \begin{scope}[xshift=8.7cm]
    \node[diagram panel title] at (2.4,2.00)
      {(b)};
    \foreach \x/\lab in {
      0/{2n+1},1.2/{2n+2},2.4/{2n+3},
      3.6/{2n+4},4.8/{2n+5}
    }{\diagramsitelabel{\x}{\lab}}
    \diagramlocaleffect{0}{e(p',g')}
    \diagrambipstate{0}{1.2}{\omega_{l_1}(\Omega_{l_1})}
    \diagrambipeffect{1.2}{2.4}{e_{s_1}}
    \diagrambipstate{2.4}{3.6}{\omega_{l_2}(\Omega_{l_2})}
    \diagrambipeffect{3.6}{4.8}{e_{s_2}}
    \diagramlocalstate{4.8}{\omega(v')}
    \foreach \x in {0,1.2,2.4,3.6,4.8}{\diagramjunction{\x}}
    \node[diagram formula] at (2.4,-2.05)
      {$\displaystyle P^2=\frac{1}{2^2}
      \left(p'+s_1s_2\,v'^\top
      \Omega_{l_2}\Omega_{l_1}g'\right)$};
  \end{scope}

  \begin{scope}[yshift=-5.0cm]
    \node[diagram panel title] at (2.5,2.00)
      {(c)};
    \foreach \x/\lab in {
      0/{2n+1},1/{2n+2},2/{2n+3},
      3/{2n+4},4/{2n+5},5/{2n+6}
    }{\diagramsitelabel{\x}{\lab}}
    \diagramlocaleffect{0}{e(p',g')}
    \diagrambipstate{0}{1}{\omega_{l_1}(\Omega_{l_1})}
    \diagrambipeffect{1}{2}{e_{s_1}}
    \diagrambipstate{2}{3}{\omega_{l_2}(\Omega_{l_2})}
    \diagrambipeffect{3}{4}{e_{s_2}}
    \diagrambipstate{4}{5}{\omega_{l_3}(\Omega_{l_3})}
    \diagramlocaleffect{5}{e(p,g)}
    \foreach \x in {0,1,2,3,4,5}{\diagramjunction{\x}}
    \node[diagram formula] at (2.5,-2.05)
      {$\displaystyle P^3=\frac{1}{2^2}
      \left(pp'+s_1s_2\,g^\top
      \Omega_{l_3}\Omega_{l_2}\Omega_{l_1}g'\right)$};
  \end{scope}

  \begin{scope}[xshift=8.7cm,yshift=-5.0cm]
    \node[diagram panel title] at (2.5,2.00)
      {(d)};
    \foreach \x/\lab in {
      0/{2n},1/{2n+1},2/{2n+2},
      3/{2n+3},4/{2n+4},5/{2n+5}
    }{\diagramsitelabel{\x}{\lab}}
    \diagramlocalstate{0}{\omega(v)}
    \diagrambipeffect{0}{1}{e_{s_1}}
    \diagrambipstate{1}{2}{\omega_{l_1}(\Omega_{l_1})}
    \diagrambipeffect{2}{3}{e_{s_2}}
    \diagrambipstate{3}{4}{\omega_{l_2}(\Omega_{l_2})}
    \diagrambipeffect{4}{5}{e_{s_3}}
    \diagramlocalstate{5}{\omega(v')}
    \foreach \x in {0,1,2,3,4,5}{\diagramjunction{\x}}
    \node[diagram formula] at (2.5,-2.05)
      {$\displaystyle P^4=\frac{1}{2^3}
      \left(1+s_1s_2s_3\,v'^\top
      \Omega_{l_2}\Omega_{l_1}v\right)$};
  \end{scope}
\end{tikzpicture}%
}
\caption{\label{fig:probability-cases}
The four possible connected closed paths in the period-two
construction computed in Lemma \ref{lem:probability-classification}:
(a) an even state to an even effect, (b) an odd state to an odd
effect, (c) an odd effect to an even effect, and (d) an even state to
an odd state.  The displayed examples have, respectively,
$L=2,2,3,$ and $2$ state bonds.  The formulas in
Lemma~\ref{lem:probability-classification} apply to arbitrary word
length $L=|w|\geq 0$, except (c), for which $L\geq1$.  Each dashed cap contributes a factor $1/2$, its
outcome contributes the displayed sign, and the state-bond matrices
ordered from left to right form
$\Omega_w=\Omega_{l_L}\cdots\Omega_{l_1}$.}
\end{figure*}

\begin{proof}
Associate a graph with the tensor network: its vertices are occurrences of elementary systems, and its edges are bipartite state generators and bipartite effect outcomes. Distinct vertices may carry the same site label. At each contracted occurrence, exactly one state leg is paired with exactly one effect leg. State edges occupy odd--even bonds, whereas effect edges occupy even--odd bonds. Consequently, at every vertex of degree two, one edge leads to the preceding site and the other to the following site. Each connected component therefore has strictly monotone site labels and is a path; neither branching nor cycles are possible.

The end points of a closed path are \chg{either}
\begin{enumerate}
\item an even local state and an even local effect,
\item an odd local effect and an odd local state,
\item an odd and an even local effect, or
\item an even and an odd local state.
\end{enumerate}
Successive contraction along the path and
Lemma~\ref{lem:directional-teleportation} give
Eqs.~\eqref{Eq:P1}--\eqref{Eq:P4}, respectively.  Tensor contractions
on distinct connected components factorize.  Finally,
coarse graining replaces elementary outcomes by sums of outcomes and
therefore produces finite sums of these products.
\end{proof}

We now show that consistency of the whole theory is equivalent to
non-negativity of the four expressions
\eqref{Eq:P1}--\eqref{Eq:P4}.

\begin{lemma}[Completion from path probabilities]
\label{lem:path-completion}
For the period-two generators in Eqs.~\eqref{Eq:E2} and
\eqref{Eq:S2}, complemented with period-two single-party GPTs, the closed generated state and effect cones form a
consistent locally tomographic multipartite GPT on every finite set
of sites if and only if the four expressions
\eqref{Eq:P1}--\eqref{Eq:P4} are non-negative for all \chg{end-point} states
and effects and all words.
\end{lemma}

\begin{proof}
Necessity is immediate: each expression
$P^1$--$P^4$ is the outcome probability of a valid finite experiment
and must therefore be non-negative in any consistent GPT. We prove
sufficiency by verifying the consistency requirements.
For every finite set of sites $\Lambda$, we must show the following:
\begin{enumerate}
\item The generated state and effect cones are mutually positive.
\item Both cones are proper and generating, and the normalized
state space is compact.
\item Product states and effects are included, and the cones are closed
under conditioning, tensor products, and marginalization.
\item The unit effect factorizes and all specified measurements are
normalized.
\item Product effects distinguish multipartite states, so that the
composite is locally tomographic.
\end{enumerate}

\emph{Mutual positivity.}
Let $\sigma$ be a generated state and let $e$ be a generated effect.
Before taking non-negative linear combinations and limits, each is
represented by a finite tensor network with some open legs. Pairing
$e$ with $\sigma$ joins these open legs and produces a closed finite
experiment. Every connected component of this closed network has one
of the four forms $P^1$--$P^4$ listed in
Lemma~\ref{lem:probability-classification}. By hypothesis, each such
factor is non-negative. Disconnected components multiply, while
coarse graining and non-negative linear combinations add with
non-negative coefficients. It follows that every such closed
network has non-negative weight. Bilinearity and continuity extend
this conclusion to the closed generated cones:
\[
e^\top\sigma\geq0
\qquad
\text{for all }
e\in C(\mathcal E_{\rm gen}^{\Lambda}),\quad
\sigma\in C(\mathcal S_{\rm gen}^{\Lambda}).
\]
This is the mutual-positivity condition.

\emph{Properness and generation.}
The local state and effect cones are generating, so they contain
spanning sets of states and effects. Their tensor products therefore
span the multipartite state and effect vector spaces on every finite
set $\Lambda$. Since the generated cones contain these product states
and effects, both generated cones are generating. They are convex
cones and are closed by construction.

It remains to prove that they are pointed. Suppose that both
$\sigma$ and $-\sigma$ belong to
$C(\mathcal S_{\rm gen}^{\Lambda})$. Mutual positivity then gives
\[
e^\top\sigma\geq0
\quad\text{and}\quad
e^\top(-\sigma)\geq0
\]
for every generated effect $e$, and hence
$e^\top\sigma=0$. In particular, $\sigma$ vanishes against every
product effect. Since the product effects span the full dual space,
this implies $\sigma=0$. Thus the generated state cone is pointed.
The same argument, with states and effects interchanged, shows that
the generated effect cone is pointed. The two cones are therefore
proper.

\emph{Operational closure.}
Product states and effects, tensor products, and marginals were
included in the construction of the generated cones. Conditioning
also preserves these cones: applying a generated effect to part of a
generated state simply joins the two corresponding finite networks.
The resulting network is therefore another generated conditional
state. The analogous statement holds for conditional effects.
Taking closures preserves all these operations because tensor
contraction is continuous.

\emph{Units and measurements.}
By construction, the unit effect on $\Lambda$ is
\[
u^\Lambda=\bigotimes_{n\in\Lambda}u^n.
\]
Each generating multipartite measurement sums to the product unit
effect on its support. The same is therefore true for its translates,
for product measurements, and for their coarse grainings. Hence all
specified measurements are normalized.

\emph{Probability bounds and compactness.}
The physical effects associated with the generated effect cone form
the order interval
\begin{equation}
\mathcal E_{\rm gen}^{\Lambda}
=
C(\mathcal E_{\rm gen}^{\Lambda})
\cap
\bigl(u^\Lambda-C(\mathcal E_{\rm gen}^{\Lambda})\bigr),
\label{eq:orderinterval}
\end{equation}
as in Eq.~\eqref{eq:physical_effects}. Thus, if $e$ is a physical
effect, both $e$ and its complementary effect $u^\Lambda-e$ belong to
the generated effect cone. Mutual positivity consequently gives
\[
0\leq e^\top\sigma
\leq (u^\Lambda)^\top\sigma=1
\]
for every normalized state $\sigma$. This proves the
upper probability bound.

The same observation shows that the unit effect is strictly
positive on every nonzero generated state. Indeed, suppose that
$\sigma\in C(\mathcal S_{\rm gen}^{\Lambda})$ satisfies
$(u^\Lambda)^\top\sigma=0$. Every product effect belonging to a
product measurement is a physical effect, so
Eq.~\eqref{eq:orderinterval} and mutual positivity imply
\[
0\leq e^\top\sigma
\leq (u^\Lambda)^\top\sigma=0.
\]
Thus $\sigma$ vanishes against all product effects. Since these effects
span the full dual space, $\sigma=0$. The generated state cone is
closed and finite dimensional, and $u^\Lambda$ is strictly positive
on its nonzero elements. Its normalized section
\[
\mathcal S_{\rm gen}^{\Lambda}
=
\left\{
\sigma\in C(\mathcal S_{\rm gen}^{\Lambda}):
(u^\Lambda)^\top\sigma=1
\right\}
\]
is therefore compact.

\emph{Local tomography.}
Finally, the multipartite vector space is the tensor product of the
local vector spaces, and the product effects span its dual. Hence two
multipartite states that give the same probabilities for every product
effect must be equal. The completed theory is therefore locally
tomographic. All the consistency requirements have now been verified,
and their only nontrivial positivity conditions are the four families
of path probabilities in
Eqs.~\eqref{Eq:P1}--\eqref{Eq:P4}.
\end{proof}

\subsection{Proof of Theorem \ref{Theo:undecidability3}}

This proof is the teleportation version of the proof of Theorem~\ref{Theo:undecidability1}. We store the matrices $\Omega_l$ in the bipartite state generators. If all their products are bounded, small hypersphere GPTs at the individual sites make all the probabilities non-negative. If the products are unbounded, any full-dimensional local GPT contains enough  states and effects for a sufficiently large product to create a negative $P^1$. Thus a consistent local completion exists exactly in the bounded case.

\begin{proof}[Proof of Theorem \ref{Theo:undecidability3}]
Let
\[
\{\Omega_l\in\mathbb{Q}^{d\times d}\}_{l=1,\ldots,k}
\]
be an instance of the matrix-product boundedness problem.  From these
matrices construct the period-two multipartite generators
of Eqs.~\eqref{Eq:E2} and \eqref{Eq:S2}.  There are two effect
generators and $k$ state generators up to translations, and all
their coordinates are rational.  It remains to prove that these
multipartite generators admit a period-two single-party completion
if and only if the products $\Omega_w$ are uniformly bounded.

\paragraph{Bounded products imply consistency.}
Suppose that there is a constant $\Lambda > 0$ such that
\[
\lVert\Omega_w\rVert\leq\Lambda
\qquad\text{for every word }w.
\]
Choose a rational $\delta>0$ with $\delta^2\Lambda\leq1$ and use, at
every site, the hypersphere GPT from
Sec. \ref{ToyExample}, with
$\epsilon=\epsilon'=\delta$.  This is a finite description and is
\chg{translation invariant} with period two.

For this \chg{hypersphere GPT}, $\lVert v\rVert,\lVert v'\rVert\leq\delta$ and the effect coordinates satisfy $p,p'\geq0$, $\lVert g\rVert\leq p\delta$, and $\lVert g'\rVert\leq p'\delta$. Hence
\begin{align*}
\lvert g^\top\Omega_wv\rvert
&\leq p\delta^2\Lambda\leq p,\\
\lvert v'^\top\Omega_wg'\rvert
&\leq p'\delta^2\Lambda\leq p',\\
\lvert g^\top\Omega_wg'\rvert
&\leq pp'\delta^2\Lambda\leq pp',\\
\lvert v'^\top\Omega_wv\rvert
&\leq\delta^2\Lambda\leq1.
\end{align*}
Equations~\eqref{Eq:P1}--\eqref{Eq:P4} are therefore non-negative.
Lemma~\ref{lem:path-completion} then shows that the multipartite
generating set is consistent.

\paragraph{Unbounded products preclude consistency.}
Conversely, suppose that the products $\Omega_w$ are unbounded and,
for contradiction, that a consistent period-two single-party
completion exists.  Apply part~(i) of
Lemma~\ref{lem:hypersphere-Balls} to the GPT at the
even sites.  It gives $v_\ast\in\mathbb{R}^d$ and
$\epsilon,\epsilon'>0$ satisfying
Eqs.~\eqref{eq:state-open-set} and \eqref{eq:effect-open-set}.

By unboundedness, choose a nonempty word $w$ such that
\[
\lVert\Omega_w\rVert>
(\epsilon\epsilon')^{-1}.
\]
There are unit vectors $n,n'$ for which
\[
n'^\top\Omega_wn=\lVert\Omega_w\rVert.
\]
Set
\begin{align*}
a&=n'^\top\Omega_wv_\ast,
\\
b&=\epsilon n'^\top\Omega_wn
=\epsilon\lVert\Omega_w\rVert.
\end{align*}
Using
\[
\max\{|a+b|,|a-b|\}\geq |b|,
\]
choose $\sigma\in\{+1,-1\}$ such that
$|a+\sigma b|\geq|b|$, and define
\[
c_\sigma
:=
\epsilon'n'^\top\Omega_w
\bigl(v_\ast+\sigma\epsilon n\bigr).
\]
Then
\begin{align}
\lvert c_\sigma\rvert
&=\epsilon'\lvert a+\sigma b\rvert
\nonumber\\
&\geq\epsilon'|b|
\nonumber\\
&=\epsilon\epsilon'\lVert\Omega_w\rVert
\nonumber\\
&>1.
\label{eq:large-correlation}
\end{align}

Use the even state
$\omega_\sigma=(1,v_\ast+\sigma\epsilon n)^\top$ and the even effect
$e=\frac12(1,\epsilon'n')^\top$.  Teleport the state along $w$ and
choose the product $s$ of the intervening outcomes to have the sign
opposite to $c_\sigma$.  Then
Eq.~\eqref{Eq:P1} gives
\begin{align}
P^1
&=
\frac{1}{2^{|w|+1}}
\left(1-\lvert c_\sigma\rvert\right)
\nonumber\\
&<0,
\end{align}
contradicting consistency.

We have thus proved that the multipartite generators admit a consistent
period-two single-party completion exactly when the semigroup
generated by $\{\Omega_l\}_{l=1}^k$ is bounded.  Since the latter
decision problem is undecidable by
Theorem~\ref{thm:unboundedness}, the claimed consistency problem is
undecidable.
\end{proof}

\subsection{Proof of Theorem \ref{Theo:undecidability4}}

The second proof similarly mirrors Theorem~\ref{Theo:undecidability2}. We now choose $\Omega_l=M_l$ from a probabilistic automaton. Stochasticity makes all matrix products bounded, so the multipartite generators are consistent with some small local GPTs. We then fix slightly larger local GPTs containing a state built from $q$ and two complementary effects built from $F$. The uniform bounds make every path probability non-negative except, potentially, the path that begins at $q$ and ends at one of these two effects. This probability is non-negative for every path exactly when the automaton never exceeds its cut point.

\begin{proof}[Proof of Theorem \ref{Theo:undecidability4}]
We reduce the strict-emptiness problem of
Theorem~\ref{thm:emptiness}, restricted to the instances described
after that theorem.  Thus we take stochastic matrices
$\{M_l\in\mathbb Q^{d\times d}\}_{l=1}^k$,
\[
q=(1,0,\ldots,0)^\top,
\qquad
F=(0,1,0,\ldots,0)^\top,
\]
which imply $F^\top q=0$, and a rational cut point
$0<\lambda<1$.

From the 
matrices $ M_l $ construct the period-two multipartite generators
of Eqs.~\eqref{Eq:E2} and \eqref{Eq:S2}, with $\Omega_l=M_l$.  There are two effect
generators and $k$ state generators up to translations, and all
their coordinates are rational. Since the matrices
$M_l$ are column stochastic,  Lemma~\ref{lemma3} gives
$\lVert M_w\rVert\leq\sqrt d$ for every word $w$.  The
bounded-products construction in the proof of
Theorem~\ref{Theo:undecidability3} therefore provides a
consistent completion starting from period-two hypersphere GPTs.

We now specify the single-party GPTs that form the other part of the
input.  Choose a rational number $\delta>0$ such that
\begin{equation}
\delta\leq\frac12,
\qquad
\delta\leq\frac{\lambda}{2\sqrt d},
\label{eq:multipartite-delta-choices}
\end{equation}
and define
\[
\omega_q:=
\begin{pmatrix}1\\q\end{pmatrix},
\qquad
e_F^\pm:=
\frac12
\begin{pmatrix}1\\ \pm F/\lambda\end{pmatrix}.
\]

At the even sites, take
\begin{align}
\mathcal{S}_0^{A^{2n}}
&=
\operatorname{conv}
\left(
\Sigma_{1,\delta}\cup\{\omega_q\}
\right),
\nonumber\\
\mathcal{E}_0^{A^{2n}}
&=
\operatorname{conv}
\left(
\Sigma_{1/2,\delta}
\cup\{0,u,e_F^+,e_F^-\}
\right).
\label{eq:theorem6-even-gpt}
\end{align}
This is
precisely the single-party GPT constructed in
Eqs.~\eqref{eq:genS0} and \eqref{eq:genE0}, with
$\epsilon=\epsilon'=\delta$.  Equation~\eqref{eq:multipartite-delta-choices}
meets the bounds used in the proof of
Theorem~\ref{Theo:undecidability2}; hence it is  a valid
single-party GPT. 

At the odd sites, take the hypersphere GPT
\begin{align}
\mathcal{S}_0^{A^{2n+1}}
&=\Sigma_{1,\delta},
\nonumber\\
\mathcal{E}_0^{A^{2n+1}}
&=
\operatorname{conv}
\left(
\Sigma_{1/2,\delta}\cup\{0,u\}
\right).
\label{eq:theorem6-odd-gpt}
\end{align}

  The resulting family of local GPTs has period two
and a finite description.

It remains to decide whether \chg{these starting single-party GPTs} are
compatible with the multipartite generators.  
From Lemma \ref{lem:path-completion}, it 
 suffices to check the probabilities listed in 
Eqs.~\eqref{Eq:P1}--\eqref{Eq:P4}. By multilinearity, it suffices to check the generating states and effects in Eqs.~\eqref{eq:theorem6-even-gpt} and \eqref{eq:theorem6-odd-gpt}. At even sites, the ball effects have $(p,g)=(1/2,x/2)$, and $e_F^\pm$ have $(p,g)=(1/2,\pm F/(2\lambda))$. At odd sites, the ball effects have $(p',g')=(1/2,x'/2)$. For Bloch vectors $x,x'$ with
$\lVert x\rVert,\lVert x'\rVert\leq\delta$,
Lemma~\ref{lemma3} implies
\begin{align}
\lvert x'^\top M_wx\rvert
&\leq\delta^2\sqrt d,
\nonumber\\
\lvert x'^\top M_wq\rvert
&\leq\delta,
\nonumber\\
\left\lvert
\frac1\lambda F^\top M_wx
\right\rvert
&\leq\frac{\delta\sqrt d}{\lambda}.
\label{eq:theorem6-uniform-correlations}
\end{align}
The choices in Eq.~\eqref{eq:multipartite-delta-choices} make all
three bounds strictly smaller than one.  These inequalities cover
every correlation in $P^2$, $P^3$, and $P^4$, and every correlation
in $P^1$ except the one using the even state $\omega_q$ and one of
the even effects $e_F^\pm$.  Null and unit \chg{end-point} effects cause no
additional condition.

For the exceptional case, Eq.~\eqref{Eq:P1} gives
\begin{equation}
P^1
=
\frac{1}{2^{L+1}}
\left(
1+\frac{s\tau}{\lambda}F^\top M_wq
\right),
\qquad
\tau\in\{+1,-1\},
\label{eq:theorem6-acceptance-probability}
\end{equation}
where $\tau$ labels the choice $e_F^\tau$ and $s$ is the product of
the intervening bipartite-effect outcomes.  Since
$F^\top M_wq\geq0$, all these probabilities are non-negative if and
only if
\[
F^\top M_wq\leq\lambda
\qquad\text{for every word }w.
\]
For the empty word this condition is satisfied because
$F^\top q=0$.  If it fails for a nonempty word, one may choose the
signs so that $s\tau=-1$, producing a negative value in
Eq.~\eqref{eq:theorem6-acceptance-probability}.

This shows that the combined
multipartite generating set is consistent exactly when  the
acceptance probabilities satisfy the \chg{cut-point} inequality $F^{\top} M_w q \leq \lambda$ for all words $w$.  This condition is the strict-emptiness problem of
Theorem~\ref{thm:emptiness}, and is therefore undecidable.
\end{proof}

\section{Conclusion}

We have shown that deciding the consistency of certain extensions of GPTs is undecidable: this holds for the addition of discrete-time dynamics, and for the addition of entangled state and effect generators in the locally tomographic, period-two translation-invariant setting studied here. These results can be compared with work showing that many problems in physics are undecidable (see, e.g., Refs.~\cite{WCPG11,C15,PCGPW25}). In the present case, the undecidable question is whether the specified data admit a consistent GPT extension.

The mechanism is the same in all four results. Although the input has a finite description, it generates processes of arbitrary length. These processes contain arbitrary products of a finite set of matrices. If a product becomes too large, or makes a probabilistic automaton cross its cut point, a suitable state and effect turn that event into a negative probability. The difficulty is therefore not checking any one experiment, but deciding whether every experiment of every finite length remains consistent.

\begin{acknowledgments} I would like to extend my deepest thanks to Elie Wolf and Maria Ciudad-Alañón for insightful discussions that initiated this research. This research was supported in part by Perimeter Institute for Theoretical Physics. Research at Perimeter Institute is supported by the Government of Canada through the Department of Innovation, Science, and Economic Development, and by the Province of Ontario through the Ministry of Colleges and Universities.  The revised version of this paper was edited with the help of ChatGPT 5.6 Sol.

\end{acknowledgments}

\section*{Data availability}
There are no publicly available research data or software supporting this article. Requests for further information or data should be sent to the author.

\end{document}